\documentclass[10pt,twocolumn,twoside]{IEEEtran}
\def\BibTeX{{\rm B\kern-.05em{\sc i\kern-.025em b}\kern-.08em
    T\kern-.1667em\lower.7ex\hbox{E}\kern-.125emX}}

\usepackage{color,url}

\usepackage{amsmath,amssymb,amsfonts}
\usepackage{graphicx}
\usepackage{textcomp}

\usepackage{amsthm}

\usepackage{enumitem}

\usepackage{bm}
\usepackage{soul}
\usepackage{comment}

\usepackage{cancel}

\usepackage[linesnumbered,ruled,vlined]{algorithm2e}
\SetKwInput{KwInput}{Input}                
\SetKwInput{KwOutput}{Output}

\newtheorem{theorem}{Theorem}

\newtheorem{corollary}{Corollary}
\newtheorem{lemma}{Lemma}
\newtheorem{assumption}{Assumption}
\newtheorem{definition}{Definition}
\newtheorem{problem}{Problem}
\newtheorem{remark}{Remark}

\allowdisplaybreaks

\usepackage{cite}
\usepackage{accents}
\usepackage{tabularx}

\begin{document}
\title{
{\fontsize{21pt}{1pt}\selectfont
A Concentration-Based Approach for Optimizing the Estimation Performance in Stochastic Sensor Selection}
}

\author{Christopher I. Calle, \textit{Student Member}, \textit{IEEE}, and  Shaunak D. Bopardikar, \textit{Senior Member}, \textit{IEEE} 
\thanks{The authors are with the Department of Electrical and Computer Engineering, Michigan State University (MSU), East Lansing, MI 48824, USA. Emails: \texttt{callechr@msu.edu, shaunak@egr.msu.edu} }
\thanks{This work was supported in part by NSF Grant \# ECCS-2030556, GAANN Grant \# P200A180025, the National GEM Fellowship, and MSU's University Enrichment Fellowship (UEF).}
}

\maketitle


\begin{abstract}
In this work, we consider a sensor selection drawn at random by a sampling with replacement policy for a linear time-invariant dynamical system subject to process and measurement noise. We employ the Kalman filter to estimate the state of the system. However, the statistical properties of the filter are not deterministic due to the stochastic selection of sensors. As a consequence, we derive concentration inequalities to bound the estimation error covariance of the Kalman filter in the semi-definite sense. Concentration inequalities provide a framework for deriving semi-definite bounds that hold in a probabilistic sense. Our main contributions are three-fold. First, we develop algorithmic tools to aid in the implementation of a matrix concentration inequality. Second, we derive concentration-based bounds for three types of stochastic selections. Third, we propose a polynomial-time procedure for finding a sampling distribution that indirectly minimizes the maximum eigenvalue of the estimation error covariance. Our proposed sampling policy is also shown to empirically outperform three other sampling policies: uniform, deterministic greedy, and randomized greedy. 

\end{abstract}



\section{Introduction}

The selection of an optimum set of inputs or outputs is a ubiquitous and challenging problem. In the field of control theory, the inputs and outputs correspond to the actuators and sensors of a dynamical system, resp., and the goal typically consists of finding a selection of actuators or sensors that optimizes an objective function and also satisfies a set of constraints and performance criteria. Refer to \cite{van2001,taha2018time,siami2020deterministic} for a few works in actuator and sensor selection in the field of control theory. In this paper, we focus on finding a selection of sensors that optimizes the estimation performance of the Kalman filter for a linear time-invariant (LTI) dynamical system.

A naive approach for finding an optimum selection of sensors is to evaluate all possible combinations. Unfortunately, since the time complexity of an exhaustive search is factorial, its use is computationally intractable for many practical problems of interest. A tractable alternative to the naive approach is a greedy sampling policy, i.e., a sequential algorithm that chooses the locally optimum choice. One advantage of greedy algorithms is the assurances that exist when the optimization problem of interest is equivalent to the maximization~\cite{krause2014submodular} or minimization~\cite{iwata2008submodular} of a submodular function. We refer the reader to \cite{clark2016submodularity} for a survey on the role of submodularity and greedy algorithms in the estimation and control of systems.

However, since many standard metrics in Kalman filtering do not exhibit submodularity as shown in \cite{jawaid2015submodularity} and \cite{zhang2017sensor}, the appealing guarantees associated with greedy algorithms, such as \cite{nemhauser1978analysis}, no longer hold. Several works attempt to justify the general utility of greedy as a benchmark heuristic in Kalman filtering. For instance, the greedy algorithm is shown to achieve near-optimal mean-square estimation error in \cite{kohara2020sensor} and \cite{hartman2019near}, and non-submodular metrics, such as the maximum eigenvalue and trace of the estimation error covariance, are also shown in \cite{chamon2020approximate} and \cite{hashemi2021randomized} to be approximately submodular. Though the classic greedy algorithm of \cite{nemhauser1978analysis} and variants of it, e.g., \cite{mirzasoleiman2015lazier} and \cite{han2021stochastic}, are of great importance due to their guarantees in submodular and even non-submodular settings, the utility of such guarantees ultimately hinges on how close to submodular the problem at hand is.

In this paper, we show that concentration inequalities (CIs), in contrast to greedy, can provide useful guarantees irrespective of submodularity for the problem of sensor selection in Kalman filtering. Though CIs offer a convenient probabilistic framework for obtaining non-trivial bounds, referred to as concentration bounds, its application in state estimation is limited and relatively unexplored. Several works that employ CIs include \cite{hashemi2021randomized}, \cite{bopardikar2019randomized}, and \cite{bopardikar2021randomized}, where concentration bounds are derived for the curvature measure in \cite{hashemi2021randomized} using the Bernstein inequality and for several observability Gramian metrics in \cite{bopardikar2019randomized} and \cite{bopardikar2021randomized} using a matrix CI, referred to as the Ahlswede-Winter (AW) inequality in \cite{bopardikar2021randomized}. The matrix CI is only made possible by the influential work \cite{ahlswede2002strong} of Ahlswede and Winter. Since \cite{bopardikar2021randomized} is the first work to apply the AW inequality in a non-trivial manner, we highlight certain aspects of it. In \cite{bopardikar2021randomized}, a Gramian lower bound is derived in terms of the Gramian of the original system for sensor placement at each location. However, since the lower bound is derived for a scaled version of the original output model, the guarantees do not hold for the observability Gramian of the original system. A major distinction between \cite{bopardikar2021randomized} and our previous work \cite{calle2021probabilistic} is that the latter work provided guarantees for the original system by recognizing the full utility of the AW inequality. Another major distinction is that \cite{calle2021probabilistic} addressed sensor selection for state estimation in the presence of process and measurement noise. In this work, we further exploit the AW inequality.







{\textit{Contributions:}} Our main contributions are three-fold. First, we develop algorithms, i.e., the first two in this work, that can find a set of parameters that satisfies the conditions required to use the AW inequality, referred to as Theorem~\ref{theorem:AW} in this work. Our algorithms aid in the implementation of the AW inequality and are novel contributions in the field of CIs.

Second, we derive novel CIs for three types of stochastic selections that we refer to as homogeneous, heterogeneous, and constrained in Section~\ref{subsection:type_of_selections}. To the best of our knowledge, we present the first CIs to bound the estimation error covariance of the Kalman filter (in the semi-definite sense) for a sensor selection drawn by a sampling \textit{with replacement} policy, i.e., a policy that draws each sensor at random and \textit{with replacement} from a pool of candidate sensors. The constrained selection accounts for sampling constraints, i.e., constraints on the number of times each candidate sensor can be chosen.


Third, we propose a procedure for finding a sampling distribution that minimizes the worst-case estimation performance of the Kalman filter. We formulate our search algorithm for the homogeneous selection, i.e., a selection that considers every candidate sensor when it is drawn by a sampling \textit{with replacement} policy. We also formulate a decentralized version of the search algorithm for the heterogeneous selection, i.e., an ensemble of mutually exclusive homogeneous selections.


In our previous work \cite{calle2021probabilistic}, we derived a CI for the steady-state error covariance of a homogeneous selection, and we proposed a procedure for finding a sampling policy that heuristically optimized the steady-state estimation performance of the Kalman filter. In this work, we derive CIs for the steady-state and filtered error covariance for the three type of stochastic selections in Section~\ref{subsection:sampling_policy}, and we propose a procedure for finding a sampling distribution that optimizes the estimation performance of the Kalman filter for the homogeneous and heterogeneous selection. We also improve the algorithm proposed in \cite{calle2021probabilistic} by significantly reducing the search space.

{\textit{Paper Organization:}}  
We outline the system model and our sampling policies in Section~\ref{section:problem_formulation}. In Section~\ref{section:main_results}, we obtain CIs for three types of stochastic selections and propose a procedure for finding a sampling distribution that optimizes estimation performance. In Section~\ref{section:simulation_results}, we offer insight into our guarantees with a numerical analysis. In Section~\ref{section:conclusion}, we summarize our findings and outline future directions of research. We provide proofs of our main results in the Appendix.

\section{Problem Formulation}
\label{section:problem_formulation}

\subsection{Notation}

We outline the convention in notation. Let $\mathbb{R}$ and $\mathbb{N}$ denote the set of real and natural numbers, resp., $I_n$ and $0_n$ denote the identity and null matrix of order $n$, resp., $\overline{\lambda} ( \cdot )$ and $\underline{\lambda}( \cdot )$ denote the maximum and minimum eigenvalue of a symmetric matrix argument, resp., and $\Delta^{n}$ denote the probability simplex in $\mathbb{R}^n$. We refer to $\mathbb{S}^{n}$, $\mathbb{S}_{+}^{n}$, and $\mathbb{S}_{++}^{n}$ as the set of symmetric, positive semi-definite (p.s.d.), and positive definite (p.d.) matrices of order $n$, resp. The operators $\succeq$ and $\succ$ hold for any matrices $A, B \in \mathbb{S}^{n}$, where the inequalities $A \succeq B$ and $A \succ B$ hold if $A-B \in \mathbb{S}_{+}^{n}$ and $A-B \in \mathbb{S}_{++}^{n}$, resp. The former and latter inequalities are said to hold in the semi-definite and definite sense, resp. Let $\{ n_1 , \ldots , n_2 \}$ denote the set of integers that include $n_1 \in \mathbb{N}$, $n_2 \in \mathbb{N}$, and the integers that range from $n_1$ to $n_2$. The notation $[n]$ is shorthand for $\{ 1 , \ldots , n \}$.

\subsection{System Model}

We refer to the sensors under consideration as \emph{candidate} sensors. The output of each of the $n_c \in \mathbb{N}$ candidate sensors is described by a linear time-invariant (LTI) model, i.e., 
\begin{align*}
y_{(t),i} = \bm{c}_i^T x_{(t)} + \bm{v}_{(t),i},
\end{align*}
where the $i$-th candidate sensor outputs the measurement $y_{(t),i} \in \mathbb{R}$ at time instant $t$. Let $x_{(t)} \in \mathbb{R}^m$ denote the latent state, $\bm{c}_i \in \mathbb{R}^m$ denote a mapping from the state $x_{(t)}$ to the uncorrupted output $\bm{c}_{i}^{T} x_{(t)}$, and $\bm{v}_{(t),i} \sim \mathcal{N}(0,\bm{\sigma}_i^2)$ denote the measurement noise. We assume $\bm{v}_{(t),i}$ is Gaussian. By assuming each candidate sensor is described by an LTI model corrupted by Gaussian noise, it is possible to identify a candidate sensor by its corresponding pair $(\bm{c}_i,\bm{\sigma}_i^2)$. To ensure that each candidate sensor in our sampling pool is \textit{unique}, we assume the pairs for any two candidate sensors are not equivalent, i.e., $(\bm{c}_i,\bm{\sigma}_i^2) \neq (\bm{c}_j,\bm{\sigma}_j^2)$ for all $i \in [ n_c ]$ and $j \in [ n_c ]/\{ i \}$. We denote the set of indices assigned to each candidate sensor as $\mathcal{I} := [n_c]$ for brevity in notation.

Though there exists $n_c$ unique candidate sensors in our sampling pool, a selection $\mathcal{S}$ of sensors can consist of multiple copies of the same candidate sensor. Let $\mathcal{S} \in \{ \hspace{0.25mm} \mathcal{I} \hspace{0.25mm} \}^{\, n_s}$ denote a sequence of $n_s$ indices. Any quantity conditioned on the selection $\mathcal{S}$ is appended by a subscript $\mathcal{S}$. This convention on notation applies to any selection introduced in this work.

We assume the system under consideration is an LTI model corrupted by process and measurement noise, i.e.,
\begin{align}
\label{eqn:system_model_lti}
\begin{split}
x_{(t+1)} &= A \, x_{(t)} + w_{(t)},   \\
y_{(t)} &= C_{\mathcal{S}} \, x_{(t)} + v_{\mathcal{S},(t)},
\end{split}
\end{align}
where $m$ and $n_s$ denote the state dimension and the number of measurements outputted by the selection~$\mathcal{S}$ of sensors, resp. Let $x_{(t)} \in \mathbb{R}^m$ and $y_{(t)} \in \mathbb{R}^{n_s}$ denote the state and output vector at time instant $t$, resp., $A \in \mathbb{R}^{m \times m}$ and $C_{\mathcal{S}} \in \mathbb{R}^{n_s \times m}$ denote the state and output matrix, resp., and $c_{i}^{T}$ denote the $i$-th row of $C_{\mathcal{S}}$. Let the initial state $x_{(0)}$, process noise $w_{(t)} \in \mathbb{R}^{m}$, and measurement noise $v_{\mathcal{S},(t)} \in \mathbb{R}^{n_s}$ denote Gaussian random variables, i.e., $x_{(0)} \sim \mathcal{N}( \overline{x}_{(0)},\Sigma_{(0)})$, $w_{(t)} \sim \mathcal{N}( 0 , Q )$, and $v_{\mathcal{S},(t)} \sim \mathcal{N}(0,R_{\mathcal{S}})$. We assume $x_{(0)}$, $\{ w_{(t)} \}_{t=0}^{\infty}$, and $\{ v_{\mathcal{S},(t)} \}_{t=0}^{\infty}$ are mutually uncorrelated. Note that the noise covariance matrices $Q$ and $R_{\mathcal{S}}$ are time-invariant.


We highlight two observations regarding the selection $\mathcal{S}$. First, the assumption that $\mathcal{S}$, i.e., the sensors corresponding to LTI model~\eqref{eqn:system_model_lti}, is time-invariant implies that the covariance matrix $R_{\mathcal{S}}$ is also time-invariant. Second, since each sensor in $\mathcal{S}$ is independent of each other, the sequence $v_{\mathcal{S},(t)}$ is uncorrelated for any time instant $t$. To clarify, $v_{\mathcal{S},(t),i} \sim \mathcal{N}( 0 , \sigma_{i}^{2} )$ is the measurement noise for the $i$-th sensor. 

The former observations imply that $R_{\mathcal{S}}$ is a diagonal matrix, where $\sigma_{i}^{2}$ is the $i$-th diagonal element of~$R_{\mathcal{S}}$. Also, the $i$-th sensor of \eqref{eqn:system_model_lti} is identified by the pair $( c_{i} , \sigma_{i}^{2} )$, where
\begin{align}
\label{eqn:pair}
( c_{i} , \sigma_{i}^{2} ) := \{ \, (\bm{c}_{j},\bm{\sigma}_{j}^2) : j = \mathcal{S}_{i} \, \}, \ \forall i \in [n_s].
\end{align}

We highlight a third observation. By assuming the variance of the measurement noise $\bm{v}_{(t),i}$ of each candidate sensor is positive, i.e., $\bm{\sigma}_i^2 > 0$ for all $i \in \mathcal{I}$, it implies the covariance matrix $R_{\mathcal{S}}$ is also p.d. 
We outline the implications from our former observations in Assumption~\ref{assumption:R}.


\begin{assumption}
\label{assumption:R}
The covariance matrix $R_{\mathcal{S}}$ of selection~$\mathcal{S}$ is p.d., i.e., $R_{\mathcal{S}} \in \mathbb{S}_{++}^{n_s}$, and diagonal.
\end{assumption}



Throughout this paper, we require the condition outlined in Assumption~\ref{assumption:01} for the proofs of several of our main results.


\begin{assumption}
\label{assumption:01}
The covariance matrix $Q \in \mathbb{S}_{++}^{m}$.
\end{assumption}

We introduce the matrix functions of Definition~\ref{def:functions} to compactly state the results of this paper.


\begin{definition}
\label{def:functions}
Assume $\Lambda, \Theta \in \mathbb{S}_{+}^{m}$, $\hat{\Lambda} \in \mathbb{S}_{++}^{m}$, $\Xi \in \mathbb{R}^{n_s \times m}$, $\Gamma \in \mathbb{S}_{+}^{n_s}$, and $\hat{\Gamma} \in \mathbb{S}_{++}^{n_s}$. Define the following matrix functions,
\begin{gather*}
f_1( \Lambda , \Xi , \Gamma ) := \Lambda - \Lambda \, \Xi^T \, ( \, \Gamma + \Xi \, \Lambda \, \Xi^T \, )^{-1} \, \Xi \, \Lambda, \\
f_2( \Lambda , \Theta ) := ( \, ( \hspace{0.25mm} A \hspace{0.25mm} \Lambda \hspace{0.25mm} A^T + Q \hspace{0.25mm} )^{-1} + \Theta \, )^{-1}, \\
f_3( \hat{\Lambda} , \Xi , \hat{\Gamma} ) :=  ( \, \hat{\Lambda}^{-1} + \Xi^T \, \hat{\Gamma}^{-1} \, \Xi \, )^{-1}, \\
f_4( \Lambda ) := A \hspace{0.25mm} \Lambda \hspace{0.25mm} A^T + Q.
\end{gather*}
\end{definition}

\vspace{-2.5mm}

\subsection{Kalman Filter}
\label{subsection:ss_KF}

We consider a Kalman filter, an optimal minimum mean-squared error (MMSE) estimator, that uses the measurements outputted by a selection~$\mathcal{S}$ of sensors. If the process and measurement noise of LTI model~\eqref{eqn:system_model_lti} do not satisfy the Gaussian assumption, then the filter is the optimal linear MMSE estimator \cite{simon2006optimal}. Let $P_{\mathcal{S},(t)}$ and $\Sigma_{\mathcal{S},(t)}$ denote the filtered and predicted covariance matrix of the state estimation error at time instant~$t$, resp. The covariance matrices $P_{\mathcal{S},(t)}$ and $\Sigma_{\mathcal{S},(t)}$ propagate in time according to the following equations,
\begin{align*}
P_{\mathcal{S},(t)} &= f_1( \Sigma_{\mathcal{S},(t)} , C_{\mathcal{S}} , R_{\mathcal{S}} ), \, \Sigma_{\mathcal{S},(t+1)} = f_4( P_{\mathcal{S},(t)} ),
\end{align*}
resp., for all $t \geq 0$. By applying Assumption~\ref{assumption:01} and assuming the output $y_{(t)}$ of LTI model~\eqref{eqn:system_model_lti} is accessible at time instant $t+1$, the above covariance equations can be stated in terms of the filtered error covariance at two different time instants $t$ and $t+1$, i.e., $P_{\mathcal{S},(t+1)} = f_2( P_{\mathcal{S},(t)} , C_{\mathcal{S}}^{T} R_{\mathcal{S}}^{-1} C_{\mathcal{S}}^{} )$.

As a consequence of Assumption~\ref{assumption:R}, the matrix $C_{\mathcal{S}}^{T} R_{\mathcal{S}}^{-1} C_{\mathcal{S}}^{}$ is equivalent to a sum of p.s.d. matrices, i.e., 
\begin{align*}
C_{\mathcal{S}}^{T} R_{\mathcal{S}}^{-1} C_{\mathcal{S}}^{} = \textstyle\sum\nolimits_{\hspace{0.25mm} i \in [n_s]} Z_{i} = \textstyle\sum\nolimits_{\hspace{0.25mm} i \in \mathcal{S}} \mathcal{Z}_{i},
\end{align*}
where the quantities $Z_i$ and $\mathcal{Z}_i$ are defined as the following,
\begin{align*}
Z_i := \sigma_{i}^{-2} c_{i}^{} c_{i}^{T}, \, \mathcal{Z}_i := \bm{\sigma}_{i}^{-2} \bm{c}_{i}^{} \bm{c}_{i}^{T}
\end{align*}
for all $i \in [n_s]$ and $i \in \mathcal{I}$, resp. The selection~$\mathcal{S}$ assigns each $Z_i$ to its corresponding $\mathcal{Z}_i$ according to \eqref{eqn:pair}.

\subsection{Sampling Policy}
\label{subsection:sampling_policy}

We employ a sampling \textit{with replacement} policy that consists of independently choosing $n_s$ sensor \textit{with replacement} from the pool of $n_c$ candidate sensors according to the sampling distribution $p$. In other words, the $i$-th candidate sensor is chosen at random with probability $p_i$ for each sampled sensor.

If a selection $\mathcal{S}$ is drawn according to our sampling \textit{with replacement} policy, then the quantity $Z_i$ in Section~\ref{subsection:ss_KF} is a random variable, i.e., sampling a sensor equates to sampling a p.s.d. matrix. Thus, the sequence $( Z_i )_{i=1}^{n_s}$ consists of $n_s$ independent and identically distributed (i.i.d.) \textit{random} variables. Let $Z$ denote any random variable in the sequence $( Z_i )_{i=1}^{n_s}$ and $\{ \mathcal{Z}_i \}_{i \in \mathcal{I}}$ denote its support. We state the following statistical properties for clarity, i.e., $\mathbb{E}[Z_i] = \mathbb{E}[Z]$ for all $i \in [n_s]$,
\begin{align*}
\mathbb{E}[Z] = \textstyle\sum\nolimits_{\, i \in \mathcal{I}} \, p_i \, \mathcal{Z}_{i}, \ \mathbb{E}[ \, C_{\mathcal{S}}^{T} R_{\mathcal{S}}^{-1} C_{\mathcal{S}}^{} \, ] = n_s \, \mathbb{E}[Z].
\end{align*}

\begin{remark}
\label{remark:IV}
The inverse transform (IT) sampling of a categorical random variable can be employed to draw a sample according to our sampling \textit{with replacement} policy. If the IT method is employed, then the computational expense of drawing a selection $\mathcal{S}$ consists of exactly $n_s$ operations. Each operation consists of sampling a uniform random variable. Thus, it is computationally inexpensive to draw a selection $\mathcal{S}$.
\end{remark}

\vspace{-2.5mm}

\subsection{Concentration Bounds}
\label{subsection:bounds}

If a selection $\mathcal{S}$ of sensors is drawn at random for the state estimation of LTI model~\eqref{eqn:system_model_lti}, then the statistical properties, i.e., the covariance equations, of the Kalman filter in Section~\ref{subsection:ss_KF} are no longer deterministic. As a consequence, we employ a variant of the Ahlswede-Winter inequality~\cite{ahlswede2002strong}, a CI that bounds a sum of i.i.d. and p.s.d. random matrices, to bound the estimation error covariance of the Kalman filter in the semi-definite sense. Corollary 2.2.2 of \cite{qiu2014sums} is stated as Theorem~\ref{theorem:AW}.


\begin{theorem}
\label{theorem:AW}
Let $( Z_i )_{i=1}^{n_s}$ denote a sequence of $n_s$ i.i.d. and p.s.d. random matrices, i.e., $Z \in \mathbb{S}_{+}^{m}$, satisfying the inequality $Z\preceq \rho \, \mathbb{E}[Z]$ almost surely for a scalar $\rho \geq 1$. If the equality
\begin{align}
\label{eqn:epsilon}
\epsilon^2 / \rho = c_0 := ( 4 / n_s ) \log_{e}{( 2m / \delta )}
\end{align}
holds for the scalars $\delta \in (0,1)$ and $\epsilon \in (0,1)$, then the event
\begin{align*}
\left\{ \, (1-\epsilon) \, n_s \, \mathbb{E}[Z] \preceq \textstyle\sum\nolimits_{\, i \in [n_s]} Z_i \preceq (1+\epsilon) \, n_s \, \mathbb{E}[Z] \, \right\}
\end{align*}
occurs at least with probability $(1-\delta)$.
\end{theorem}



We define the quantity $\mathcal{T}$ as the the set consisting of the following conditions, $\rho \geq 1$, $\epsilon \in (0,1)$, and $p \in \Delta^{n_c}$, i.e., 
\begin{align}
\label{def:T}
\mathcal{T} := \{ \, \rho \geq 1 , \, \epsilon \in (0,1) , \, p \in \Delta^{n_c} \, \}.
\end{align}
Also, throughout this paper, the following inequalities,
\begin{align}
\label{eqn:Z_inequality}
\mathcal{Z}_j \preceq \rho \, \textstyle\sum\nolimits_{\, i \in \mathcal{I}} \, p_i \, \mathcal{Z}_i, \, \forall j \in \mathcal{I}.
\end{align}
are satisfied in place of satisfying $Z\preceq \rho \, \mathbb{E}[Z]$ almost surely since they are equivalent. We now identify regimes of the parameters for which the conditions of Theorem~\ref{theorem:AW} hold.


\begin{lemma}
\label{lemma:AW_conditions}
If the inequality $\varrho^{*} < c_0^{-1}$ holds, where $\varrho^{*}$ is the optimal value of the following convex program,
\begin{align}
\label{alg:varrho}
\min_{ \varrho \, \geq \, 1 , \, p \, \in \, \Delta^{n_c} } \, \varrho \ \, \mathrm{s.t.} \,
\begin{bmatrix}
\, \mathbb{E}[ Z ] & \bm{c}_j \, \, \\
\, \bm{c}_j^{T} & \rho \, \bm{\sigma}_j^2 \, \, 
\end{bmatrix}
\succeq 0, \ \forall j \in \mathcal{I},
\end{align}
then for any scalar $\epsilon \in [ \sqrt{\varrho^{*} c_0} , 1 )$ there exists a $\rho \in [ \varrho^{*} , c_0^{-1} )$ and a $p \in \Delta^{n_c}$ that satisfy the conditions \eqref{eqn:epsilon} and \eqref{eqn:Z_inequality}.
\end{lemma}

\begin{proof}
In this proof, a feasible pair $(\rho,p)$ for Theorem~\ref{theorem:AW} is shown to exist given any $\epsilon$ on the interval $[ \sqrt{\varrho^{*} c_0} , 1 )$.


First, we derive an upper bound on $\rho$. Note that $\epsilon < 1$ and \eqref{eqn:epsilon} imply $\rho < c_0^{-1}$. Second, we compute a lower bound on $\rho$. Note that for an arbitrary $p \in \Delta^{n_c}$, there exists a minimum $\rho$ that satisfies \eqref{eqn:Z_inequality}. Note also that \eqref{eqn:Z_inequality} is equivalent to the $n_c$ LMI constraints in \eqref{alg:varrho}. By solving the convex program in Lemma~\ref{lemma:AW_conditions}, we find the minimum value of $\rho$, denoted as $\varrho^{*}$, that satisfies the conditions required to use Theorem~\ref{theorem:AW}. Note that $\varrho^{*} \leq \rho$ and $\rho < c_0^{-1}$ imply $\varrho^{*} < c_0^{-1}$. Thus, $\rho$ exists on the interval $[ \varrho^{*} , c_0^{-1} )$. Third, we derive a lower bound on $\epsilon$. Note that $\varrho^{*} \leq \rho$ and $\epsilon^2 = \rho \, c_0$ imply $\epsilon \geq \sqrt{\varrho^{*} c_0}$ $\Rightarrow$ $\epsilon > 0$ since $\varrho^{*} \geq 1$ and $c_0 > 0$. Thus, if a $\rho \in [ \varrho^{*} , c_0^{-1} )$ exists, then an $\epsilon \in [ \sqrt{\varrho^{*} c_0} , 1 )$ that satisfies the equality $\epsilon^2 = \rho \, c_0 \in (0,1)$ also exists. Note that a unique $\epsilon$ exists for every $\rho$. 

Similarly, if $\varrho^{*} < c_0^{-1}$, then an $\epsilon$ exists on the interval $[ \sqrt{\varrho^{*} c_0} , 1 )$. Thus, if $\epsilon \in [ \sqrt{\varrho^{*} c_0} , 1 )$ exists, then a $\rho \in [ \varrho^{*} , c_0^{-1} )$ that satisfies the equality $\epsilon^2 = \rho \, c_0 \in (0,1)$ also exists. 
\end{proof}

\begin{remark}
If $p$ is a fixed quantity, then $\varrho^{*} < c_0^{-1}$ if and only if $n_s > \underbar{n}_s := 4 \, \varrho^{*} \log{\frac{2m}{\delta}}$. Thus, Lemma~\ref{lemma:AW_conditions} also holds if $n_s > \underbar{n}_s$. The proof consists of applying the definition of $c_0$ and rearranging the inequality $\varrho^{*} < c_0^{-1}$ of Lemma~\ref{lemma:AW_conditions}.
\end{remark}

\textsc{Algorithm~\ref{alg:02}} outlines the procedure for finding a feasible $n_s$ when $p$ is given. In contrast, \textsc{Algorithm~\ref{alg:03}} outlines a general procedure for finding a feasible $p$ when $n_s$ is given. Both algorithms output quantities that satisfy the conditions required to use Theorem~\ref{theorem:AW}. Recall that $c_0$ is defined in \eqref{eqn:epsilon} and $\varrho^{*}$ is computed by executing the program~\eqref{alg:varrho} of Lemma~\ref{lemma:AW_conditions}.

\vspace{-0.5mm}

\begin{algorithm}
\DontPrintSemicolon
  \KwInput{$p$, $\delta$, $\{ \mathcal{Z}_i \}_{i \in \mathcal{I}}$}
  
  Compute $\varrho^{*}$; Choose $n_s \in \mathbb{N} : n_s > 4 \varrho^{*} \log{\frac{2m}{\delta}}$
  
  Select $\epsilon \in [ \sqrt{\varrho^{*} c_0} , 1 ) \subset (0,1)$; Compute $\rho : \rho = \epsilon^2 c_0^{-1}$
  
  \KwOutput{$n_s, \epsilon, \rho$}

\caption{Selection of the Sample Size}
\label{alg:02}
\end{algorithm}

\vspace{-7.5mm}

\begin{algorithm}
\DontPrintSemicolon
  \KwInput{$n_s$, $\delta$, $\{ \mathcal{Z}_i \}_{i \in \mathcal{I}}$}
  
  Compute $\varrho^{*}$, $c_0$; Select $\epsilon \in [ \sqrt{\varrho^{*} c_0} , 1 ) \subset (0,1)$
  
  Compute $\rho : \rho = \epsilon^2 c_0^{-1}$; Choose $p \in \Delta^{n_c} : \eqref{eqn:Z_inequality}$
  
  \KwOutput{$\epsilon, \rho, p$}

\caption{Selection of the Sampling Distribution}
\label{alg:03}
\end{algorithm}

\vspace{-3.5mm}

\subsection{Selection Types}
\label{subsection:type_of_selections}
Throughout this paper, we consider three types of selections. For a selection of the homogeneous type, every one of the $n_c$ candidate sensors is considered each time a sensor is chosen for that selection. Unless otherwise stated, a selection $\mathcal{S}$ of sensors is referred to as homogeneous, i.e., a selection drawn by \textsc{Algorithm~\ref{alg:homogeneous_SS}}, from this point onward.

\vspace{-1.5mm}

\begin{algorithm}
\DontPrintSemicolon
  \KwInput{$n_s$, $p$}

  Draw a selection $\mathcal{S}$ of $n_s$ sensors per the sampling \textit{with replacement} policy outlined in Section~\ref{subsection:sampling_policy}
  
  \KwOutput{$\mathcal{S}$}

\caption{Homogeneous Sensor Selection}
\label{alg:homogeneous_SS}
\end{algorithm}

\vspace{-2.5mm}

For a selection $\mathcal{H}$ of the heterogeneous type, the original sampling pool of $n_c$ candidate sensors is segmented into $K$ mutually exclusive partitions and a homogeneous selection $\mathcal{S}^{(i)}$ is drawn from each of the $i \in [K]$ partitions according to \textsc{Algorithm~\ref{alg:homogeneous_SS}}, i.e., each partition samples ${n_s}^{(i)}$ sensors \textit{with replacement} from their sampling pool of ${n_c}^{(i)}$ candidate sensors per the sampling distribution $p^{(i)} \in \Delta^{{n_c}^{(i)}}$. By fusing the stochastic selections drawn by each partition, we form a heterogeneous ensemble $\mathcal{H} := \bigcup_{i \in [K]} \mathcal{S}^{(i)}$ of sensors.

When each partition draws a homogeneous selection, the selection $\mathcal{S}^{(i)}$ is drawn from the set $\{ \hspace{0.25mm} \mathcal{I}^{\, (i)} \hspace{0.05mm} \}^{n_s^{(i)}}$, where $\mathcal{I}^{\, (i)} := \{ \phi^{(i)} ,\ldots , \varphi^{(i)} \} \subset \mathcal{I}$ denotes the collection of indices (of the candidate sensors) assigned to each partition, and
\begin{align*}
\phi^{(i)} := 1 + \textstyle\sum_{j=1}^{\, i} {n_c}^{(j-1)}, \, \varphi^{(i)} := \textstyle\sum_{j=1}^{\, i} {n_c}^{(j)}
\end{align*}
denote the initial and final index of the candidate sensors assigned to the $i$-th partition, resp. Note that ${n_c}^{(0)} = 0$ since ${n_c}^{(i)}$ is only physically defined for $i \in [K]$.


In the homogeneous setup, recall that random variable~$Z$ is assigned to the unpartitioned sampling pool of candidate sensors. In the heterogeneous setup, i.e., the partitioned case, the $i$-th partition is assigned the random variable $Z^{(i)}$, a random p.s.d. matrix. Let $\{ \mathcal{Z}_{i} \}_{i \in \mathcal{I}^{(i)}}$ denote~its support and $\mathbb{E}[Z^{(i)}]$ denotes its expectation, i.e.,
\begin{align*}
\mathbb{E}[Z^{(i)}] = \textstyle\sum\nolimits_{\, ( j , k ) \, \in \, ( \, [{n_c}^{(i)}] \, , \, \mathcal{I}^{(i)} \, )} \, {p_j}^{(i)} \, \mathcal{Z}_k.
\end{align*}


For a selection $\mathcal{C}$ of the constrained type, a limit $k_i$ is imposed on the number of times the $i$-th candidate sensor can be drawn \textit{with replacement}. Let $k \in \{ \mathbb{N} , 0 \}^{n_c}$ denote the list of constraints imposed on each of the $n_c$ candidate sensors.


\begin{definition}
\label{def:I}
Let the binomial random variable $I_i$ denote the number of times the $i$-th candidate sensor is chosen $\textit{with replacement}$ after $n_s$ sampling instances from the pool of $n_c$ candidate sensors.
\end{definition}

Let $\{ I_i \leq k_i \}$ denote the event that the $i$-th candidate sensor is chosen no more than $k_i$ times after $n_s$ independent trials. As a consequence of Definition~\ref{def:I}, the probability of selecting the $i$-th candidate sensor $k_i$ times after $n_s$ independent trials is equal to the following, 
\begin{align}
\label{eqn:C}
\mathbb{P}[ I_i = k_i ] = f_5 ( n_s, k_i, p_i ) := \binom{n_s}{k_i} {p_i}^{k_i} (1-{p_i})^{n_s-k_i}.
\end{align}

\begin{definition}
Let $\mathcal{K}$ denote the event that the $i$-th candidate sensor is chosen no more than $k_i$ times for all $ i \in [n_c]$, i.e., 
\begin{align*}
\mathcal{K} := \textstyle\bigcap_{\, i \in [n_c]} \{ I_i \leq k_i \},
\end{align*}
and $\mathcal{K}^c$ denote its complement, i.e., $\mathcal{K}^{c} := \textstyle\bigcup_{\, i \in [n_c]} \{ I_i \nleq k_i \}$.
\end{definition}

If a selection $\mathcal{S}$ satisfies the sampling constraints specified by the vector $k$, then the event~$\mathcal{K}$ is said to have occurred. In Section~\ref{subsection:constrained_SS}, we assume the following conditions.

\begin{assumption}
\label{assumption:k}
Define the quantities $k_m := \textstyle\max_{\, i \in [n_c]} \, k_i$ and $k_{\Sigma} := \textstyle\sum\nolimits_{\, i \in \mathcal{I}} \, k_i$. Assume $n_s \in [ k_m , k_{\Sigma} ]$ and the following,
\begin{align}
\label{eqn:k}
k_i = 0 \ \mathit{if} \ p_i = 0 \ \mathit{and} \ k_i \in \{ \mathbb{N} , 0 \} \ \mathit{if} \ p_i \neq 0, \ \forall i \in \mathcal{I}.
\end{align}
\end{assumption}

We comment on the conditions outlined in Assumption~\ref{assumption:k}. First, we assume $n_s \geq k_m$ to guarantee the binomial coefficient of the probability~\eqref{eqn:C} exists. Second, we assume $n_s \leq k_{\Sigma}$ to guarantee the event~$\mathcal{K}$ is physically realizable. Third, we assume \eqref{eqn:k} to guarantee that the sequence $k$ is well-defined.


\begin{definition}
\label{definition:RV_N}
Let $N$ denote the number of selections drawn~by \textsc{Algorithm~\ref{alg:constrained_SS}} prior to drawing a constrained selection $\mathcal{C}$.
\end{definition}

Any selection drawn by \textsc{Algorithm~\ref{alg:constrained_SS}} is referred to as constrained. Before drawing a selection~$\mathcal{C}$ that realizes event~$\mathcal{K}$, \textsc{Algorithm~\ref{alg:constrained_SS}} draws $N \in \mathbb{N}$ selections (independent of each other) since \textsc{Algorithm~\ref{alg:homogeneous_SS}} is not guaranteed to draw a homogeneous selection that also satisfies the sampling constraints specified by $k$. Since \textsc{Algorithm~\ref{alg:constrained_SS}} must verify that a homogeneous selection $\mathcal{S}$ satisfies its sampling constraints, the operator $\cancel{\circ}$ is introduced, where $\mathcal{S} \, \cancel{\circ} \, k$ means that the selection~$\mathcal{S}$ does not satisfy the sampling constraints of $k$.

\begin{algorithm}
\DontPrintSemicolon
  \KwInput{$k$, $n_s$, $p$}
  
  Initialize $\mathcal{C} \leftarrow [ \, ]$, $N \leftarrow 0$
  
  \While{$\mathcal{C} \, \cancel{\circ} \, k$ \textbf{or} $\mathcal{C} = \emptyset$}
  {
  
  Execute \textsc{Algorithm~\ref{alg:homogeneous_SS}}; Set $\mathcal{C} \hspace{-0.55mm} \leftarrow \hspace{-0.55mm} \mathcal{S}$; Set $N \hspace{-0.55mm} \leftarrow N \hspace{-0.55mm} + \hspace{-0.55mm} 1$

  }

  \KwOutput{$\mathcal{C}$, $N$}

\caption{Constrained Sensor Selection}
\label{alg:constrained_SS}
\end{algorithm}

\subsection{Problem Statement}
\label{subsection:problem_statement}

In this paper, we study the statistical properties of the Kalman filter for each selection in Section~\ref{subsection:type_of_selections}.

In Problem~\ref{problem:bounds_S}, the goal is to determine whether guarantees on estimation performance exist when we randomly choose a selection of sensors according to \textsc{Algorithm~\ref{alg:homogeneous_SS}}. We provide a solution to Problem~\ref{problem:bounds_S} in Section~\ref{subsection:homogeneous_SS} for the time-dependent and steady-state error covariance. Our solution is extended to a heterogeneous selection in Section~\ref{subsection:heterogeneous_SS}.

\begin{problem}
\label{problem:bounds_S}
Obtain concentration bounds on the estimation performance, i.e., the estimation error covariance of the Kalman filter, for a homogeneous selection $\mathcal{S}$ of sensors.
\end{problem}


In Problem~\ref{problem:bounds_C}, we quantify how the presence of sampling constraints affects our solution to Problem~\ref{problem:bounds_S}. We provide a solution in Section~\ref{subsection:constrained_SS} for the steady-state error covariance.

\begin{problem}
\label{problem:bounds_C}
Generalize the concentration bounds of a homogeneous selection $\mathcal{S}$ to account for sampling constraints.
\end{problem}

In Problem~\ref{problem:algo}, we find a sampling \textit{with} replacement policy that optimizes our ability to perform state estimation. We specify our sampling with replacement policy by specifying the sampling distribution $p$ of Section~\ref{subsection:sampling_policy}. 
We provide a solution to Problem~\ref{problem:algo} in Section~\ref{subsection:proposed}.

\begin{problem}
\label{problem:algo}
Find a sampling distribution that can draw a sensor selection that optimizes estimation performance.
\end{problem}

\section{Main Results}
\label{section:main_results}

Our main results consist of novel concentration bounds on the error covariance of the Kalman filter for three types of stochastic selections and an algorithm for finding a sampling policy that is optimal with respect to estimation performance.

\subsection{Homogeneous Sensor Selection}
\label{subsection:homogeneous_SS}

First, we quantify the estimation performance that we achieve in the probabilistic sense for a homogeneous selection~$\mathcal{S}$ of sensors. We formulate a CI for the filtered error covariance $P_{\mathcal{S},(t)}$ in Theorem~\ref{theorem:SS_bounds_t}.

\begin{theorem}
\label{theorem:SS_bounds_t}
Let $\Sigma_{(t)} \in \mathbb{S}_{++}^{m}$ denote the predicted error covariance at time instant $t$. Define the quantities,
\begin{gather*}
U_{S,(t)} := ( \, \Sigma_{(t)}^{-1} + (1-\epsilon) \, n_s \, \mathbb{E}[Z] \, )^{-1} \in \mathbb{S}_{++}^{m},   \\
L_{S,(t)} := ( \, \Sigma_{(t)}^{-1} + (1+\epsilon) \, n_s \, \mathbb{E}[Z] \, )^{-1} \in \mathbb{S}_{++}^{m},
\end{gather*}
for the time instant $t$. Let $P_{\mathcal{S},(t)} \in \mathbb{S}_{++}^{m}$ denote the filtered error covariance of selection $\mathcal{S}$ at time instant $t$, such that
\begin{align}
\label{eqn:P_S_t_02}
P_{\mathcal{S},(t)} &= ( \, \Sigma_{(t)}^{-1} + C_{\mathcal{S}}^{T} R_{\mathcal{S}}^{-1} C_{\mathcal{S}}^{} \, )^{-1}.
\end{align}
Then, the event $\{ \hspace{0.25mm} L_{S,(t)} \preceq P_{\mathcal{S},(t)} \preceq U_{S,(t)} \hspace{0.25mm} \}$ at time instant $t$ holds at least with probability $(1-\delta)$, i.e.,
\begin{align}
\label{eqn:concentration_inequality}
\mathbb{P} [ \, L_{S,(t)} \preceq P_{\mathcal{S},(t)} \preceq U_{S,(t)} \, ] \geq (1 - \delta).
\end{align}
\end{theorem}

We comment on the appeal of semi-definite guarantees.


\begin{remark}
\label{remark:bounds}
%
%
%
%
%
Bounds (in the semi-definite sense) on the estimation error covariance are appealing because they imply assurances on standard eigenvalue-based metrics in Kalman filtering, e.g., the maximum eigenvalue, trace, determinant, and condition number of the estimation error covariance. The fact that every CI in this work guarantees semi-definite bounds in the probabilistic sense is why our CIs are noteworthy.
\end{remark}


In Theorem~\ref{theorem:SS_bounds}, we formulate the steady-state performance of the time-dependent quantities in Theorem~\ref{theorem:SS_bounds_t}.

\begin{theorem}
\label{theorem:SS_bounds}
Assume $(A, \mathbb{E}[Z]^{1/2})$ and $(A,C_{\mathcal{S}})$ are detectable for any selection $\mathcal{S} \in \mathcal{I}^{\, n_s}$. Let $P_{\mathcal{S}}$ denote the steady-state error covariance of selection $\mathcal{S}$ and the matrices $U_{S}$, $L_{S}$, and $P_{\mathcal{S}}$ denote the unique p.d. steady-state solution to
\begin{gather*}
U_{S,(t+1)} = f_2( \, U_{S,(t)} \, , \, (1-\epsilon) \, n_s \, \mathbb{E}[Z] \, ),   \\
L_{S,(t+1)} = f_2( \, L_{S,(t)} \, , \, (1+\epsilon) \, n_s \, \mathbb{E}[Z] \, ), \\
P_{\mathcal{S},(t+1)} = f_2( \, P_{\mathcal{S},(t)} \, , \, C_{\mathcal{S}}^{T} R_{\mathcal{S}}^{-1} C_{\mathcal{S}}^{} \, ),
\end{gather*}
resp., such that $U_{S,(0)}, L_{S,(0)}, P_{\mathcal{S},(0)} \in \mathbb{S}_{++}^{m}$. Then,
\begin{align}
\label{eqn:CI_SS}
\mathbb{P} [ \, L_S \preceq P_{\mathcal{S}} \preceq U_S \, ] \geq (1 - \delta).
\end{align}
\end{theorem}

\begin{remark}
\label{remark:detectability_assumptions}
The detectability assumptions of Theorem~\ref{theorem:SS_bounds} hold for any of the following sufficient conditions:
\begin{enumerate}[ label=(\roman*) , leftmargin=6.5mm ]
    \item $(A,\bm{c}_i)$ is detectable for all $i \in \mathcal{I}$, or
    \item Before randomly sampling $n_s$ sensors via \textsc{Algorithm~\ref{alg:homogeneous_SS}}, first choose $n_a$ sensors that guarantee LTI system~\eqref{eqn:system_model_lti} is detectable, i.e., strategically draw a selection $\mathcal{S}_a \in \mathcal{I}^{\, n_a}$ that guarantees the pair $(A,C_{\mathcal{S}_a})$ is detectable.
\end{enumerate}
If condition (ii) holds, then substitute the selection $\mathcal{S}$ in Theorem~\ref{theorem:SS_bounds} for the selection $\mathcal{S} \cup \mathcal{S}_a$, i.e., a union of the randomly sampled $\mathcal{S}$ and deterministically chosen $\mathcal{S}_a$.
\end{remark}


Next, we discuss how to select the quantities stated in the above theorems. Since the above theorems are derived using Theorem~\ref{theorem:AW}, the quantities $\epsilon$ and $\delta$ need to satisfy the conditions of Theorem~\ref{theorem:AW}. To find a feasible set of parameters that satisfies the conditions of Theorem~\ref{theorem:AW}, we can employ \textsc{Algorithm~\ref{alg:02}}. 




\subsection{Heterogeneous Sensor Selection}
\label{subsection:heterogeneous_SS}

Next, we quantify the estimation performance that we achieve in the probabilistic sense for a heterogeneous selection~$\mathcal{H}$ of sensors. We formulate a CI for the steady-state error covariance~$P_{\mathcal{H}}$ in Corollary~\ref{corollary:SS_bounds_K}. Although similar guarantees also exist for the filtered error covariance $P_{\mathcal{H},(t)}$ by applying Theorem~\ref{theorem:SS_bounds_t}, we omit them for brevity.

\begin{corollary}
\label{corollary:SS_bounds_K}
Suppose the following conditions,
\begin{enumerate}[ label=(\roman*) , leftmargin=6.5mm ]
    \item $(A, \mathbb{E}[Z^{(i)}]^{1/2})$ is detectable,
    \item $(A,C_{\mathcal{S}^{(i)}})$ is detectable for any selection $\mathcal{S}^{(i)} \! \in \! \{ \mathcal{I}^{(i)} \hspace{-0.25mm} \}^{{n_s}^{(i)}}$\!\!.
\end{enumerate}
for all $i \in [K]$. Let $P_{\mathcal{H}}$ denote the steady-state error covariance of selection $\mathcal{H}$ and the matrices $U_{H}$, $L_{H}$, and $P_{\mathcal{H}}$ denote the unique p.d. steady-state solution to
\begin{gather}
U_{H,(t+1)} = f_2( \, U_{H,(t)} \, , \, \textstyle\sum\nolimits_{\, i \in [K]} \, (1-\epsilon^{(i)}) \, {n_s}^{(i)} \, \mathbb{E}[Z^{(i)}] \, ),   \nonumber \\
L_{H,(t+1)} = f_2( \, L_{H,(t)} \, , \, \textstyle\sum\nolimits_{\, i \in [K]} \, (1+\epsilon^{(i)}) \, {n_s}^{(i)} \, \mathbb{E}[Z^{(i)}] \, ),   \nonumber \\
\label{eqn:P_H_t_01}
P_{\mathcal{H},(t+1)} = f_2( \, P_{\mathcal{H},(t)} \, , \, C_{\mathcal{H}}^{T} R_{\mathcal{H}}^{-1} C_{\mathcal{H}}^{} \, ),
\end{gather}
resp., such that $U_{H,(0)}, L_{H,(0)}, P_{\mathcal{H},(0)} \in \mathbb{S}_{++}^{m}$. Then,
\begin{align*}
\mathbb{P} [ \, L_H \preceq P_{\mathcal{H}} \preceq U_H \, ] \geq \textstyle\prod\nolimits_{\, i \in [K]} \, ( 1 - \delta^{(i)} ).
\end{align*}
\end{corollary}

\begin{proof}
First, observe that the guarantees of Theorem~\ref{theorem:SS_bounds} also apply to each partition in a heterogeneous setup. For example, if the requirements of Theorem~\ref{theorem:SS_bounds} are met, then the following CI holds for the $i$-th partition,
\begin{equation}
\begin{aligned}
\label{eqn:sum_Z_inequality_i}
& \ \ \ \ \ \ \mathbb{P} \big{[} \, (1-\epsilon^{(i)}) \, {n_s}^{(i)} \, \mathbb{E}[Z^{(i)}] \preceq C_{\mathcal{S}^{(i)}}^{T} R_{\mathcal{S}^{(i)}}^{-1} C_{\mathcal{S}^{(i)}}^{}   \\
&= \textstyle\sum_{j=1}^{{n_s}^{(i)}} \! \! {Z_{j}}^{(i)} \! \preceq (1+\epsilon^{(i)}) \, {n_s}^{(i)} \, \mathbb{E}[Z^{(i)}] \, \big{]} \geq ( 1 - \delta^{(i)} ).
\end{aligned}
\end{equation}
Since each partition is mutually exclusive, the CIs of each partition can be aggregated. Observe that $C_{\mathcal{H}}^{T} R_{\mathcal{H}}^{-1} C_{\mathcal{H}^{}} = \sum_{i=1}^{K} C_{\mathcal{S}^{(i)}}^{T} R_{\mathcal{S}^{(i)}}^{-1} C_{\mathcal{S}^{(i)}}^{}$. Thus, we obtain
\begin{equation}
\begin{aligned}
\label{eqn:sum_Z_inequality_sum_i}
& \, \mathbb{P} \big{[} \, \textstyle\sum_{i=1}^{K} (1-\epsilon^{(i)}) \, {n_s}^{(i)} \, \mathbb{E}[Z^{(i)}] \preceq \textstyle\sum_{i=1}^{K} \textstyle\sum_{j=1}^{{n_s}^{(i)}} {Z_{j}}^{(i)}   \\
&= C_{\mathcal{H}}^{T} R_{\mathcal{H}}^{-1} C_{\mathcal{H}^{}} \preceq \textstyle\sum_{i=1}^{K} (1+\epsilon^{(i)}) \, {n_s}^{(i)} \, \mathbb{E}[Z^{(i)}] \, \big{]} \geq \Psi,
\end{aligned}
\end{equation}
where $\Psi := \prod_{i=1}^{K} \, ( 1 - \delta^{(i)} )$. Finally, we derive the assumptions and quantities of Corollary~\ref{corollary:SS_bounds_K} by following a derivation similar to that of Theorem~\ref{theorem:SS_bounds}.
\end{proof}

We comment on certain aspect of Corollary~\ref{corollary:SS_bounds_K}.

\begin{remark}
The conditions of Remark~\ref{remark:detectability_assumptions} are also applicable to each partition under consideration by Corollary~\ref{corollary:SS_bounds_K}.
\end{remark}

\begin{remark}
In the unpartitioned case, i.e., when the number of partitions $K$ is equal to unity, the guarantees of Corollary~\ref{corollary:SS_bounds_K} reduce to those of Theorem~\ref{theorem:SS_bounds}, a special case.
\end{remark}

\begin{remark}
\label{remark:compare}
The guarantees on estimation performance by Corollary~\ref{corollary:SS_bounds_K} and Theorem~\ref{theorem:SS_bounds} can be properly compared for any number of partitions $K$ if the following conditions hold,
\begin{align*}
n_s = \textstyle\sum\limits_{i \in [K]} {n_s}^{(i)}, \, n_c = \textstyle\sum\limits_{i \in [K]} {n_c}^{(i)}, \, (1-\delta) = \textstyle\prod\limits_{i \in [K]} (1-\delta^{(i)}).
\end{align*}
The first (second) condition ensures that the homogeneous and heterogeneous guarantees are considering the same number of sampled (candidate) sensors. The third condition ensures that the probabilities associated with the events of Corollary~\ref{corollary:SS_bounds_K} and Theorem~\ref{theorem:SS_bounds} are identical.
\end{remark}

%
%
%
%
%
%


Next, we discuss how to select the quantities stated in the above corollary. To find a feasible set of parameters for each of the $K$ partitions under consideration by Corollary~\ref{corollary:SS_bounds_K}, we can employ \textsc{Algorithm~\ref{alg:02}} to each partition.


\subsection{Constrained Sensor Selection}
\label{subsection:constrained_SS}


Next, we study how the CIs of Section~\ref{subsection:homogeneous_SS} are affected when subjected to sampling constraints. Though the CIs in this section are formulated for the steady-state error covariance, they can also be extended to the filtered error covariance. 

We first introduce the function $\Phi$ to ensure that any quantity $\mathbb{P}[ \cdot ]$ is bounded by scalars that lie on the interval $[0,1]$. We repeatedly employ $\Phi$ and $f_5$ of \eqref{eqn:C} for compact notation.


\begin{definition}
\label{def:phi}
Assume $\gamma \in (-\infty,1]$. Define the function
\begin{equation*}
\Phi ( \gamma ) := 
    \begin{cases}
        \, \gamma & \text{if } \, \gamma \in [0,1] \\
        \, 0 & \text{if } \, \gamma < 0
    \end{cases}.
\end{equation*}
\end{definition}


Next, we generalize the guarantees of Theorem~\ref{theorem:SS_bounds} to consider any homogeneous selection $\mathcal{S}$ irrespective of whether or not the sampling constraints under consideration are satisfied.

\begin{theorem}
\label{theorem:intersection_bounds}
Define $\alpha := 1 - [ \, n_c - \textstyle\sum_{j=1}^{n_c} \textstyle\sum_{i=0}^{k_{j}} f_5 ( n_s, i, p_j ) \, ]$, a scalar that does not exceed unity, i.e., $\alpha \leq 1$. Then, the events, $\{ L_{S} \preceq P_{\mathcal{S}} \preceq U_{S} \}$ and $\mathcal{K}$, hold simultaneously for a selection $\mathcal{S}$ at least with probability $\Phi ( \alpha - \delta )$, i.e.,
\begin{align}
\label{eqn:CI-C}
\mathbb{P}[ \, \{ L_{S} \preceq P_{\mathcal{S}} \preceq U_{S} \} \, \cap \, \mathcal{K} \, ] \geq \Phi ( \alpha - \delta ).
\end{align}
\end{theorem}

\begin{proof}
First, we define $\mathcal{A} := \{ L_{S} \preceq P_{\mathcal{S}} \preceq U_{S} \}$ and derive a lower bound on $\mathbb{P}[ \mathcal{K} ]$ in Lemma~\ref{lemma:event_bound}. We refer the reader to the Appendix for the proof of Lemma~\ref{lemma:event_bound}. 

\vspace{-1mm}

\begin{lemma}
\label{lemma:event_bound}
Event $\mathcal{K}$ satisfies the inequality $\mathbb{P}[ \mathcal{K} ] \geq \Phi ( \alpha )$.
\end{lemma}

\vspace{-1mm}

Next, if we assume $\mathbb{P}[ \mathcal{A} ] \geq (1-\delta)$, then $\mathbb{P}[ \mathcal{A} \cap \mathcal{K} ] \geq ( \mathbb{P}[ \mathcal{K} ] - \delta )$ since the quantity $\mathbb{P}[ \mathcal{A} ]$ is equal to $\mathbb{P}[ \mathcal{A} \cap \mathcal{K} ] + 1 - \mathbb{P}[ \mathcal{K} ] - \mathbb{P}[ \mathcal{A}^{c} \cap \mathcal{K}^{c} ]$. We finally obtain \eqref{eqn:CI-C} by applying Lemma~\ref{lemma:event_bound} to the inequality $\mathbb{P}[ \mathcal{A} \cap \mathcal{K} ] \geq ( \mathbb{P}[ \mathcal{K} ] - \delta )$.
\end{proof}

We comment on the significance of $\alpha$.

\begin{remark}
\label{remark:special_case}
Though no sampling constraints are explicitly defined for a homogeneous selection, it is implicitly assumed by \textsc{Algorithm~\ref{alg:homogeneous_SS}} that $k_i = n_s$ for all $i \in [n_c]$. If $k_i = n_s$ for all $i \in [n_c]$, then $\alpha$ is equal to unity, i.e., $\alpha = 1$, and the theorems in this section reduce to Theorem~\ref{theorem:SS_bounds}, a special case.
\end{remark}


\begin{remark}
\label{remark:uniform}
The evaluation of $\alpha$ is computationally expensive for small $n_s$ and large $n_c$. A remedy to decrease the cost of evaluating $\alpha$ is to assume a uniform set of sampling constraints, i.e., $k_i = k_u \in \{ \mathbb{N} , 0 \}$ for all $i \in [n_c]$, where the constant $k_u$, referred to as the uniformity factor, is sufficient information for specifying the sampling constraints of $k$.
\end{remark}


Next, we generalize the guarantees of Theorem~\ref{theorem:SS_bounds} to consider a homogeneous selection that also satisfies the sampling constraints under consideration, i.e., a constrained selection.

\begin{theorem}
\label{theorem:conditional_bounds}
The expected number $N$ of selections drawn (by \textsc{Algorithm~\ref{alg:constrained_SS}}), before sampling the first constrained selection $\mathcal{C}$, is at most $\alpha^{-1}$, i.e., $\mathbb{E}[ N ] \leq \alpha^{-1}$, and
\begin{align}
\label{eqn:CI-C2}
\mathbb{P}[ \, L_{C} \preceq P_{\mathcal{C}} \preceq U_{C} \, ] \geq \Phi ( 1 - \delta / \alpha ).
\end{align}
\end{theorem}

\begin{proof}
First, we derive the inequality $\mathbb{E}[ N ] \leq \alpha^{-1}$. We define the following scalars, $\gamma := \mathbb{P}[ \mathcal{K} ]$ and $\zeta := (1-\gamma)$, and the infinite sum $S := \sum_{i=1}^{\infty} \zeta^i = \zeta \hspace{0.5mm} (1-\zeta)^{-1}$. Observe that
\begin{align}
\label{eqn:EN}
\mathbb{E}[ N ] = \textstyle\sum_{n=1}^{\infty} n \hspace{0.5mm} \mathbb{P}[ N=n ] = \textstyle\sum_{n=1}^{\infty} n \hspace{0.5mm} \gamma \hspace{0.5mm} (1-\gamma)^{n-1}.
\end{align}
Also, observe that the infinite sum $S$ is a convergent power series since $\gamma \in [0,1]$. Its derivative is equal to the following,
\begin{align}
\label{eqn:dS}
dS \hspace{0.25mm} / \hspace{0.25mm} d\zeta = \textstyle\sum_{i=1}^{\infty} i \hspace{0.5mm} \zeta^{i-1} = (1-\zeta)^{-2} = \gamma^{-2}.
\end{align}
Observe that \eqref{eqn:EN} and \eqref{eqn:dS} imply $\mathbb{E}[ N ] = \gamma^{-1}$. We finally obtain the inequality $\mathbb{E}[ N ] \leq \alpha^{-1}$ by applying the inequality $\mathbb{P}[ \mathcal{K} ] =: \gamma \geq \alpha$ of Lemma~\ref{lemma:event_bound} to the equality $\mathbb{E}[ N ] = \gamma^{-1}$.

Next, we derive \eqref{theorem:conditional_bounds}. We begin at the point in the proof of Theorem~\ref{theorem:intersection_bounds} before applying Lemma~\ref{lemma:event_bound}, i.e., $\mathbb{P}[ \mathcal{A} \cap \mathcal{K} ] \geq ( \mathbb{P}[ \mathcal{K} ] - \delta )$. We obtain \eqref{theorem:conditional_bounds} by executing the following steps,
\begin{align*}
\mathbb{P}[ \mathcal{A} \cap \mathcal{K} ] \geq ( \mathbb{P}[ \mathcal{K} ] - \delta ) \, \overset{(a)}{\Leftrightarrow} \, \mathbb{P}[ \mathcal{A} | \mathcal{K} ] \geq ( 1 - \delta / \mathbb{P}[ \mathcal{K} ] ) \, \overset{(b)}{\Rightarrow} \eqref{eqn:CI-C2}.
\end{align*}
where step~(a) applies Bayes' theorem and step~(b) employs Lemma~\ref{lemma:event_bound}. Observe that $\mathcal{A} | \mathcal{K} := \{ L_{C} \preceq P_{\mathcal{C}} \preceq U_{C} \}$.
\end{proof}


The distinction between Theorem~\ref{theorem:intersection_bounds} and \ref{theorem:conditional_bounds} is that the former is a guarantee for any drawn homogeneous selection $\mathcal{S}$ and the latter is one for any drawn constrained selection $\mathcal{C}$. Recall that the $N$-th selection drawn by \textsc{Algorithm~\ref{alg:constrained_SS}} is the first of $N$ draws to satisfy the sampling constraints specified by $k$.

\begin{remark}
\label{remark:verification}
Besides sampling, \textsc{Algorithm~\ref{alg:constrained_SS}} also verifies that each of the $N$ drawn selections meets its sampling constraints. Note that the verification step of \textsc{Algorithm~\ref{alg:constrained_SS}} is computationally inexpensive according to Remark~\ref{remark:IV}.
\end{remark}

\begin{remark}
\label{remark:inequality}
The probability that the event~$\{ \, L_{C} \preceq P_{\mathcal{C}} \preceq U_{C} \, \}$ occurs is at least equal to or greater than that of the event $\{ L_{S} \preceq P_{\mathcal{S}} \preceq U_{S} \} \cap \mathcal{K}$ since $\Phi ( 1 - \delta / \alpha ) \geq \Phi ( \alpha - \delta )$.
\end{remark}

If $\alpha^{-1}$ is not very large, then Remark~\ref{remark:verification} and Remark~\ref{remark:inequality} suggest that the CI of Theorem~\ref{theorem:conditional_bounds} is preferred over that of Theorem~\ref{theorem:intersection_bounds} since $\mathbb{E}[ N ] \leq \alpha^{-1}$.



We now highlight a connection between sampling \textit{with} and \textit{without} replacement. First, we should clarify that the selection drawn via a sampling \textit{without} replacement policy from set~$\mathcal{I}$ is a subset selection, i.e., a selection where no candidate sensor is chosen more than once. Next, if a constrained selection~$\mathcal{C}$ is drawn, where $n_s \leq n_c$ and $k_i = 1$ for all $i \in \mathcal{I}$, and at least $n_s$ elements of $p$ are non-zero, then $\mathcal{C}$ is also a subset selection. For this special case, our sampling \textit{with} replacement policy and a sampling \textit{without} replacement policy are comparable in the sense that they draw a subset selection.


\subsection{Proposed Sampling Distribution}
\label{subsection:proposed}

In this section, we study how to find a sampling distribution that optimizes estimation performance. First, we outline the optimization problems that we want to solve for the purpose of optimizing the estimation performance of the Kalman filter. We choose the worst-case estimation performance, i.e., the maximum eigenvalue of the error covariance, to gauge the quality of state estimation for our Kalman filter.

Since the quantities $\overline{\lambda}(P_{\mathcal{S},(t)})$ and $\overline{\lambda}(P_{\mathcal{S}})$ of interest cannot be directly minimized due to the stochastic nature of the estimation error covariance, we indirectly minimize them by minimizing their upper bounds $\overline{\lambda}(U_{S,(t)})$ and $\overline{\lambda}(U_{S})$, resp. Note that minimizing the upper bounds $\overline{\lambda}(U_{S,(t)})$ and $\overline{\lambda}(U_{S})$ is equivalent to solving the following optimization problems,
\begin{equation}
\label{eqn:opt_03}
\begin{aligned}
\min_{ U_{S,(t)}, \gamma } \hspace{1.0mm} & \overline{\lambda} \hspace{0.25mm} ( \, U_{S,(t)} \, )   \\
\mathrm{s.t.} \hspace{3.0mm}
& \eqref{def:T}, \hspace{0.25mm} U_{S,(t)} \hspace{-1.0mm} \in \hspace{-0.25mm} \mathbb{S}_{++}^{m}, \hspace{0.25mm} \eqref{eqn:epsilon}, \hspace{0.25mm} \eqref{eqn:Z_inequality}, \hspace{0.25mm} U_{S,(t)} \hspace{-0.5mm} = \hspace{-0.5mm} f_2( \Sigma_{(t)} , \Pi \hspace{0.25mm} ),
\end{aligned}
\end{equation}
\begin{equation}
\label{eqn:opt_01}
\begin{aligned}
\min_{ U_S, \gamma } \hspace{2.5mm} & \overline{\lambda} \hspace{0.25mm} ( \, U_S \, )   \\
\mathrm{s.t.} \hspace{2.5mm}
& \eqref{def:T}, \hspace{0.25mm} U_S \in \mathbb{S}_{++}^{m}, \hspace{0.25mm} \eqref{eqn:epsilon}, \hspace{0.25mm} \eqref{eqn:Z_inequality}, \hspace{0.25mm} U_S = f_2( \hspace{0.25mm} U_S , \Pi \hspace{0.25mm} ),
\end{aligned}
\end{equation}
resp., where $\gamma := \{ \rho, \epsilon, p \}$ and $\Pi := (1-\epsilon) \, n_s \, \mathbb{E}[Z]$. 

When the quantities $\epsilon$ and $\rho$ are fixed, the following are convex programs to solving the aforementioned problems,
\begin{equation}
\label{eqn:opt_04}
\begin{aligned}
\max_{ \lambda , p , X_{(t)} } \hspace{0.25mm} & \lambda   \\
\mathrm{s.t.} \hspace{2.5mm}
& \lambda > 0, \, p \in \Delta^{n_c}, \, X_{(t)} \in \mathbb{S}_{++}^{m}, \, X_{(t)} \succeq \lambda I_m, \, \eqref{eqn:Z_inequality},   \\
&
X_{(t)} = (A \, \Sigma_{(t)} A^{T} + Q)^{-1} + (1-\epsilon) \, n_s \, \mathbb{E}[Z],
\end{aligned}
\end{equation}
\begin{equation}
\label{eqn:opt_02}
\begin{aligned}
\max_{ \lambda , p , X } \quad &   \lambda   \\
\mathrm{s.t.} \hspace{4.5mm}
& \lambda > 0, \, p \in \Delta^{n_c}, \, X \in \mathbb{S}_{++}^{m}, \, X \succeq \lambda I_m, \, \eqref{eqn:Z_inequality}, \\
&
\begin{bmatrix}
( - X + Q^{-1} + \Pi ) & ( Q^{-1} A )   \\
( Q^{-1} A )^{T} & ( X + A^{T} Q^{-1} A )
\end{bmatrix}
\succeq 0,
\end{aligned}
\end{equation}
resp., where \eqref{eqn:opt_04} is a convex formulation of \eqref{eqn:opt_03} and \eqref{eqn:opt_02} is a convex relaxation of \eqref{eqn:opt_01}. Note that the pair $(A, \mathbb{E}[Z]^{1/2})$ must be detectable for the programs \eqref{eqn:opt_01} and \eqref{eqn:opt_02} to output an upper bound $U_S$ on the steady-state error covariance $P_S$.

Next, we comment on the convex programs. First, we refer the reader to \textsc{Algorithm}~\ref{alg:03} on how to select~$\epsilon$ and compute~$\rho$. Second, the solution $p^{*}$ to \eqref{eqn:opt_04} is the sampling distribution that minimizes $\overline{\lambda} \hspace{0.25mm} ( \hspace{0.25mm} U_{S,(t)} \hspace{0.25mm} )$. Third, the solution $p^{*}$ to \eqref{eqn:opt_02} is the optimal sampling distribution of a convex relaxation of \eqref{eqn:opt_03}. Thus, the convex program \eqref{eqn:opt_02} is only a heuristic for minimizing the steady-state quantity $\overline{\lambda}( \hspace{0.25mm} U_{S} \hspace{0.25mm} )$. Fourth, since \eqref{eqn:opt_04} and \eqref{eqn:opt_02} are convex, they can be solved in polynomial-time using interior-point methods \cite{vandenberghe2005interior}. The complexity with respect to the number of floating point operations (flops) to execute a general interior-point algorithm is $O( m^3 n_c^3 )$. The complexity is dominated by \eqref{eqn:Z_inequality}, a constraint in \eqref{eqn:opt_04} and \eqref{eqn:opt_02} that consists of $n_c$ semi-definite inequalities.







Next, we show that \eqref{eqn:opt_03} is equivalent to \eqref{eqn:opt_04}. First, we define $X_{(t)} \hspace{-0.25mm} := \hspace{-0.25mm} U_{S,(t)}^{-1}$ for clarity in notation. Next, we reformulate \eqref{eqn:opt_03} as the maximization of $\underline{\lambda}( X_{(t)} )$. As a consequence of the reformulation, the constraint $X_{(t)} \succeq \lambda I_m$ is added to \eqref{eqn:opt_02}. Next, we obtain the final constraint of \eqref{eqn:opt_04} by stating the final constraint of \eqref{eqn:opt_03} in terms of $X_{(t)}$. Next, observe that the constraint $X_{(t)} \in \mathbb{S}_{++}^{m}$ always holds since $\Sigma_{(t)} \in \mathbb{S}_{++}^{m}$ and $(1-\epsilon) \, n_s \, \mathbb{E}[Z] \in \mathbb{S}_{+}^{m}$. Also, observe that the following constraints, $X_{(t)} \in \mathbb{S}_{++}^{m}$ and $X_{(t)} \succeq \lambda I_m$, imply $\lambda > 0$ since the objective is to maximize $\lambda$. Thus, the following constraints, $X_{(t)} \in \mathbb{S}_{++}^{m}$ and $\lambda > 0$, are redundant and only stated in \eqref{eqn:opt_04} for clarity. By following a similar derivation and relaxing the equality constraint $U_S = f_2( \hspace{0.25mm} U_S , \Pi \hspace{0.25mm} )$, we can show that \eqref{eqn:opt_02} is a convex relaxation of \eqref{eqn:opt_01}. The proof also consists of applying the definition $X:=U_{S}^{-1}$, the matrix inversion lemma, and the Schur complement method.

We employ \textsc{Algorithm~\ref{alg:04}}, a polynomial-time procedure that executes a grid search over $\epsilon$, to find a sampling distribution $p^{*}$, i.e., a sampling \textit{with replacement} policy, that minimizes the time-dependent quantity~$\overline{ \lambda}( \hspace{0.25mm} U_{S,(t)} )$.

\begin{remark}
Though \textsc{Algorithm~\ref{alg:03}} is a general procedure for finding a feasible sampling distribution, it does not explicitly outline how to find one. \textsc{Algorithm~\ref{alg:04}} is established to simultaneously execute \textsc{Algorithm~\ref{alg:03}} and find a sampling distribution that is optimal with respect to state estimation.
\end{remark}

\begin{remark}
\label{remark:algo}
In step~5 of \textsc{Algorithm~\ref{alg:04}}, if \eqref{eqn:opt_02} is executed instead of \eqref{eqn:opt_04}, then the algorithm outputs a proposed sampling distribution $p^{*}$ for the steady-state quantity $\overline{\lambda}(U_{S})$.
\end{remark}


\vspace{-2.5mm}

\begin{algorithm}
\DontPrintSemicolon
  \KwInput{$n_s$, $n_p$, $\delta$, $\{ \mathcal{Z}_{i} \}_{i \in \mathcal{I}}$}
  
  Compute $c_0$, $\varrho^{*}$

  Generate a list $\Phi$ of $n_p$ equally-spaced points on the interval $[ \sqrt{\varrho^{*} c_0} , 1 )$, i.e., a grid space of $n_p$ points
  
  \For{$i \in [n_p]$}
  {
  Compute $\hat{\rho}$ : $\hat{\epsilon} := \Phi[i]$, $\hat{\rho} = \hat{\epsilon}^2 c_0^{-1}$
  
  Execute the program \eqref{eqn:opt_04} for the pair $( \hat{\epsilon} , \hat{\rho} )$
    
  Record the tuple $( \hat{\epsilon}, \hat{\rho}, p^{*}, \lambda^{*} )_{i}$ of the $i$-th grid point
  
  }

  Select the tuple corresponding to the minimum ${\lambda^{*}}^{-1}$
  
  Redefine $\hat{\epsilon}$, $\hat{\rho}$, and $p^{*}$ as the first three quantities, resp., corresponding to the selected tuple

  \KwOutput{$\hat{\epsilon}$, $\hat{\rho}$, $p^{*}$}

\caption{Proposed Sampling Distribution}
\label{alg:04}
\end{algorithm}

\vspace{-2.5mm}



Next, we formulate a search procedure for finding a set $\{ \hspace{0.25mm} {p^{*}}^{(i)} \hspace{0.25mm} \}_{i \in [K]}$ of proposed sampling distributions in the heterogeneous setup. The procedure consists of applying \textsc{Algorithm~\ref{alg:04}} to each of the $K$ partitions, i.e., finding a sampling distribution that optimizes estimation performance for each partition, and returning $\{ \hat{\epsilon}^{(i)}, \hat{\rho}^{(i)}, {p^{*}}^{(i)} \}$ for the $i$-th partition. We refer to this procedure as the heterogeneous version of \textsc{Algorithm}~\ref{alg:04}. Though this decentralized approach does not optimally improve the estimation performance of the Kalman filter, it does reduce the computational cost of executing \textsc{Algorithm}~\ref{alg:04} for the homogeneous setup since it converts one large optimization problem into $K$ smaller ones.



\section{Numerics}
\label{section:simulation_results}

In Section~\ref{section:simulation_01}, we compare the estimation performance of the sampling policy proposed in Section~\ref{subsection:proposed} against three other sampling policies: uniform, deterministic greedy, and randomized greedy. In Section~\ref{section:simulation_02}, we study a decentralized version of \textsc{Algorithm~\ref{alg:04}}. In Section~\ref{section:simulation_03}, we study how $\delta$ and our proposed sampling distribution affect $\alpha$. We focus our analysis on the steady-state estimation performance.

Throughout this section, we assume $m = 3$, $\delta = 0.05$, and $Q = 0.50 \, I_m$. Each element of the state matrix $A$  is independently chosen at random from a uniform distribution. Each element of the sequence $\bm{c}_i$ is similarly chosen for each candidate sensor. We assume the measurement noise variance of each candidate sensor is identical, i.e., $\bm{\sigma}_{i}^{2} = 0.50$ for all $i \in \mathcal{I}$. We also verify that each pair $(A,\bm{c}_{i})$ is detectable in order to satisfy the detectability conditions of Theorem~\ref{theorem:SS_bounds}. As suggested in Section~\ref{subsection:ss_KF}, the set $\{ \mathcal{Z}_i \}_{i \in \mathcal{I}}$ of matrices is computed from the set $\{ (\bm{c}_i,\bm{\sigma}_{i}^{2}) \}_{i \in \mathcal{I}}$ of pairs.

\subsection{Homogeneous Sensor Selection}
\label{section:simulation_01}

In this section, we assume $n_c=420$. For convenience, we also assume $n_p=5$, i.e., a coarse grid space.

In Figure~\ref{figure:homogeneous_ss}a, we compare the steady-state estimation performance of the following sampling policies, randomized greedy, deterministic greedy, and the sampling distribution proposed by \textsc{Algorithm~\ref{alg:04}}, for a range of sampled sensors. Figure~\ref{figure:homogeneous_ss}a shows that the estimation performance $\overline{\lambda}( P_\mathcal{S} )$ of our proposed sampling distribution tends to on-average outperform the estimation performance of both greedy sampling policies. Figure~\ref{figure:homogeneous_ss}a also suggests that sampling a selection from a few candidate sensors, as opposed to many, can result in superior estimation performance since the sampling distributions proposed by \textsc{Algorithm~\ref{alg:04}} tend to be sparse. Figure~\ref{figure:homogeneous_ss}a also shows that the upper and lower bounds on $\overline{\lambda}( P_{\mathcal{S}} )$, denoted as $\overline{\lambda}( U_S )$ and $\overline{\lambda}( L_S )$, resp., become significantly tighter as the number of sampled sensors increases.

For the greedy sampling policy, we compute the steady-state quantity $\overline{\lambda}( P_{\mathcal{G}} )$ of a greedy selection~$\mathcal{G}$ drawn according to \textsc{Algorithm~\ref{alg:10}}\footnote{First, the sensor selection outputted by \textsc{Algorithm}~\ref{alg:10} is randomly chosen if $\gamma \in [0,1)$ and deterministically chosen if $\gamma = 1$. The greedy approaches employed in this work sample \textit{with replacement} in order to be able to properly compare their estimation performance against our proposed sampling policy. In Figure~\ref{figure:homogeneous_ss} and \ref{figure:homogeneous_ss_greedy_time}, we assume $\gamma = 0.10$ for the randomized greedy sampling policy. Second, in step 3 of \textsc{Algorithm}~\ref{alg:10}, the elements of $\mathcal{I}_s$ are sampled uniformly and \textit{without} replacement from $\mathcal{I}$. Third, we employ the notation $[\mathcal{G};e]$ to indicate that the element $e$ is appended to the sequence $\mathcal{G}$.}. Let $P_{\mathcal{G}} \in \mathbb{S}_{++}^{m}$ denote the steady-state error covariance of selection $\mathcal{G}$ and the steady-state solution to $P_{\mathcal{G},(t+1)} = f_2( \, P_{\mathcal{G},(t)} \, , \, C_{\mathcal{G}}^{T} R_{\mathcal{G}}^{-1} C_{\mathcal{G}}^{} \, )$ given $P_{\mathcal{G},(0)} \in \mathbb{S}_{++}^{m}$.

For the proposed sampling distribution, we run $100$ Monte Carlo trials to approximate the mean and standard deviation of the steady-state estimation performance $\overline{\lambda}( P_{\mathcal{S}} )$ for a range of sampled sensors. The bounds $\overline{\lambda}( U_S )$ and $\overline{\lambda}( L_S )$ are computed for the same range of sampled sensors.

\vspace{-2mm}

\begin{algorithm}
\DontPrintSemicolon
  \KwInput{$\gamma$, $n_s$, $\{ \mathcal{Z}_i \}_{i \in \mathcal{I}}$}
  
  Initialize $\mathcal{G} \leftarrow [ \, ]$
  
  \For{$i \in [n_s]$ }
  {

  Randomly choose a subset $\mathcal{I}_s$ of $\mathcal{I}$ : $| \hspace{0.25mm} \mathcal{I}_s | = \lceil \gamma \hspace{0.10mm} n_c \rceil$ 
  
  Compute $g^{*} := \arg \min\nolimits_{\, g \in \mathcal{I}_s} \overline{\lambda} \, ( \, P_{ \, [ \hspace{0.25mm} \mathcal{G} \hspace{0.25mm} ; \hspace{0.25mm} g \hspace{0.25mm} ]} \, )$
  
  Set $\mathcal{G} \leftarrow [ \, \mathcal{G} \, ; \, g^{*} \, ]$
  }
  
  \KwOutput{$\mathcal{G}$}

\caption{Greedy Sensor Selection}
\label{alg:10}
\end{algorithm}

\vspace{-2.5mm}

In Figure~\ref{figure:homogeneous_ss}b, we compare the steady-state estimation performance of two sampling policies, a uniform sampling distribution $u := \{ \hspace{0.25mm} 1/n_c \hspace{0.25mm} \}^{n_c}$, i.e., each candidate sensor is chosen with equal probability, and the sampling distribution proposed by \textsc{Algorithm~\ref{alg:04}}, for a range of sampled sensors. Figure~\ref{figure:homogeneous_ss}b shows that the proposed sampling distribution requires significantly fewer sampled sensors to achieve the same estimation performance. Thus, we show that the uniform distribution is not an appealing sampling policy.

For the proposed sampling policy, the steady-state estimation performance $\overline{\lambda}( U_{S} )$ for a given $n_s$ is computed using the sampling distribution proposed by \textsc{Algorithm~\ref{alg:04}}.

For the uniform sampling policy, the minimum steady-state estimation performance $\overline{\lambda}( U_{U} )$ for a given $n_s$ is computed by substituting the quantity $\epsilon$ in the equations \eqref{eqn:uniform_U} and \eqref{eqn:uniform_L} for $\epsilon_{u}^{*} := 4 \, \varrho_{u}^{*} \, \log{ ( \hspace{0.25mm} 2m / \delta \hspace{0.25mm} ) }$, where
\vspace{-0.5mm}
\begin{align}
\label{prog:u}
\varrho_{u}^{*} := &\arg \min_{\varrho \geq 1} \ \varrho \ \ \mathrm{s.t.} \ \eqref{eqn:Z_inequality}, \, p = u := \{ \hspace{0.25mm} 1/n_c \hspace{0.25mm} \}^{n_c}.
\end{align}
\vspace{-4mm}

Let $\mathcal{U}$ denote a selection drawn via a uniform sampling policy and $P_{\mathcal{U},(t)}$$( P_{\mathcal{U}} )$ denote the filtered (steady-state) error covariance for the uniform selection $\mathcal{U}$. Let $U_{U,(t)}$$( U_{U} )$ and $L_{U,(t)}$$( L_{U} )$ denote the time-dependent (steady-state) bounds of $P_{\mathcal{U},(t)}$$( P_{\mathcal{U}} )$, where $U_{U}$ and $L_{U}$ are the unique p.d. steady-state solutions to the following,
\vspace{-0.5mm}
\begin{gather}
\label{eqn:uniform_U}
U_{U,(t+1)} = f_2( \, U_{U,(t)} \, , \, (1-\epsilon) \, n_s \textstyle\sum\nolimits_{\, i \in [n_c]} u_i \mathcal{Z}_i \, ), \\
\label{eqn:uniform_L}
L_{U,(t+1)} = f_2( \, L_{U,(t)} \, , \, (1+\epsilon) \, n_s \textstyle\sum\nolimits_{\, i \in [n_c]} u_i \, \mathcal{Z}_i \, ),
\end{gather}
\vspace{-0.5mm}
resp., such that $U_{U,(0)}, L_{U,(0)} \in \mathbb{S}_{++}^{m}$. 

In Figure~\ref{figure:homogeneous_ss_greedy_time}a, we show the average time required to execute \textsc{Algorithm~\ref{alg:04}} for a range of candidate sensors. Since the number of sampled sensors does not affect the computational complexity of the optimization programs in Section~\ref{subsection:proposed}, the execution times plotted in Figure~\ref{figure:homogeneous_ss_greedy_time}a hold for any $n_s$.

In Figure~\ref{figure:homogeneous_ss_greedy_time}b, we compare the average time\footnote{Due to the parallel nature of \textsc{Algorithm}~\ref{alg:04}, the average execution time in Figure~\ref{figure:homogeneous_ss_greedy_time}a and Figure~\ref{figure:heterogeneous_ss}b is also averaged over the number of grid points.} required to draw a randomized and deterministic greedy selection according to \textsc{Algorithm~\ref{alg:10}}. Since randomized greedy only considers a fraction of the candidate sensors at each sampling instant, it is expected to outperform its deterministic counterpart in execution time. Figure~\ref{figure:homogeneous_ss_greedy_time}b confirms our expectations.


In the regime of large $n_s / n_c$, the time require to execute \textsc{Algorithm~\ref{alg:04}} is comparable in execution time to both greedy approaches. For large $n_s / n_c$, our algorithm is preferable over greedy approaches since the execution time of \textsc{Algorithm~\ref{alg:04}} is independent of the value of $n_s$.

\begin{figure}
    \centering
    \includegraphics[width=\columnwidth]{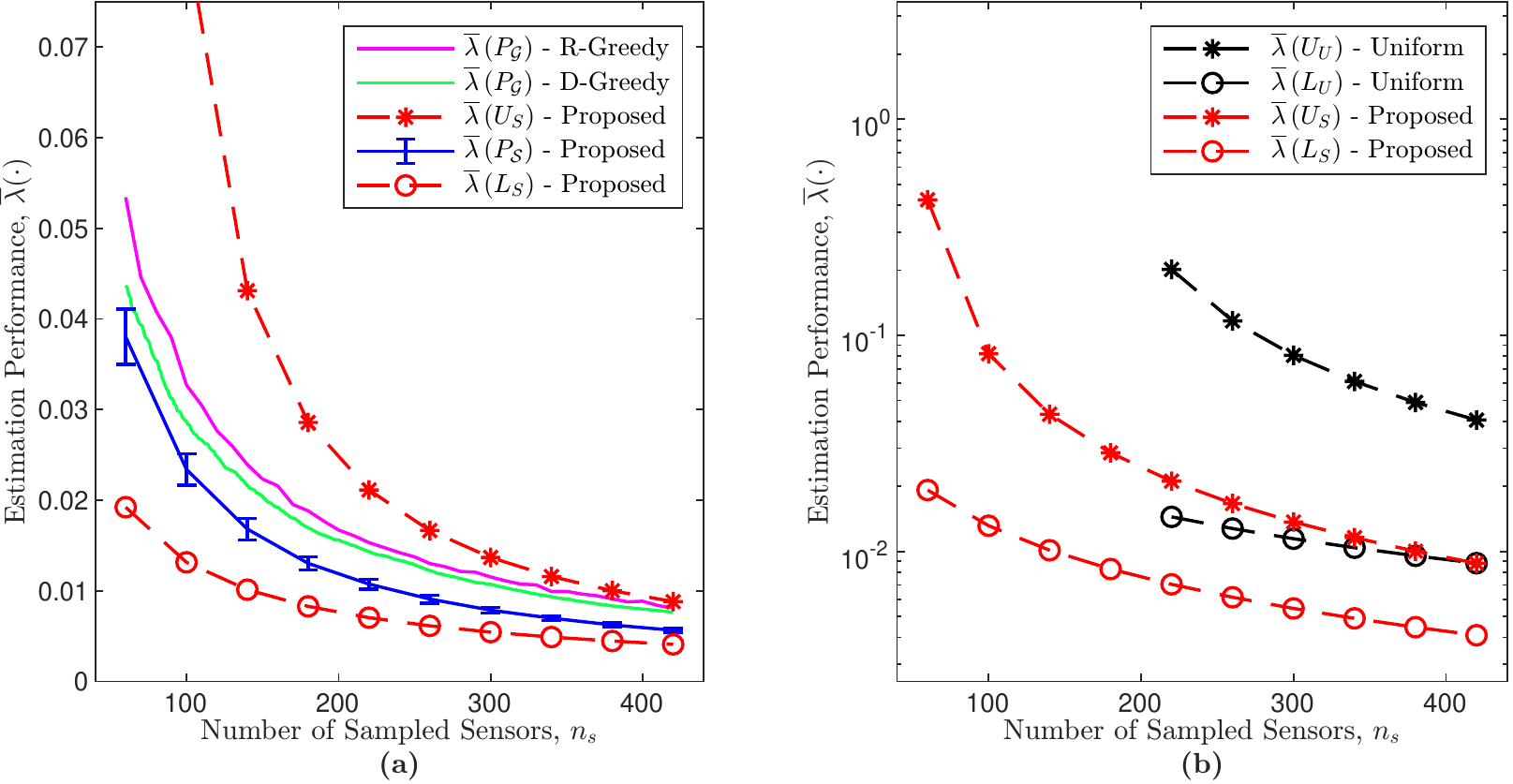}
    \caption{(a) Comparison of the estimation performance of randomized greedy (R-Greedy), deterministic greedy (D-Greedy), and our proposed sampling policy for varying $n_s$. Only the average value of $\overline{\lambda}( P_G )$ for the randomized greedy sampling policy is shown. The average value of $\overline{\lambda}( P_\mathcal{S} )$ is indicated by the (blue-line) curve and the variability of $\overline{\lambda}( P_\mathcal{S} )$ is captured by the standard deviation, where the error bars indicate $\pm$ one standard deviation. (b)~Comparison of the estimation performance of a uniform (black-star) and our proposed (red-star) sampling policy for varying $n_s$.
    }
    \label{figure:homogeneous_ss}
\end{figure}

\begin{figure}
    \centering
    \includegraphics[width=\columnwidth]{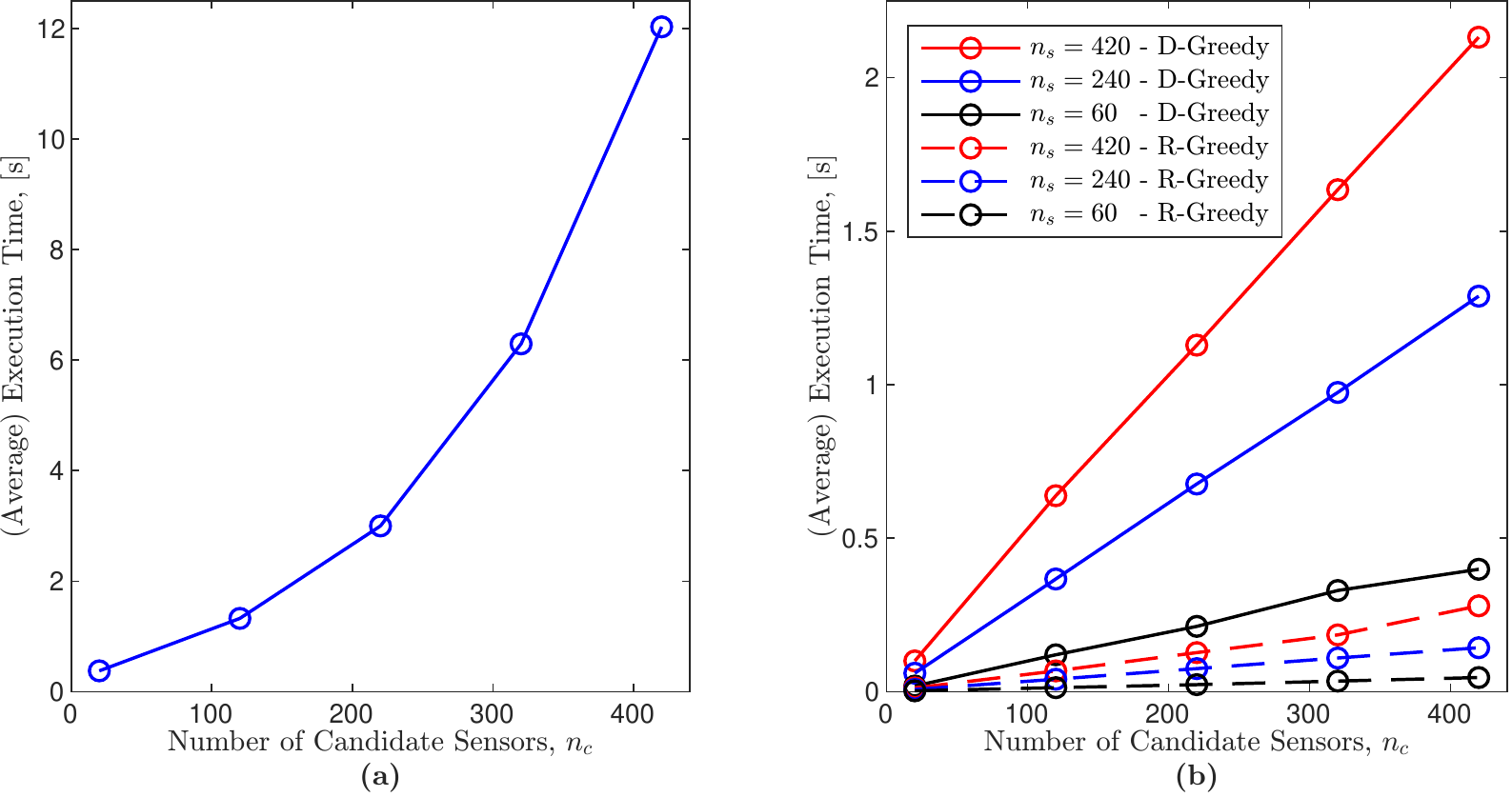}
    \caption{(a) Plot of the average time required to execute the steady-state equivalent of \textsc{Algorithm~\ref{alg:04}} for varying $n_c$. (b) Plot of the average time required to execute \textsc{Algorithm~\ref{alg:10}} for the deterministic greedy (D-Greedy) and randomized greedy (R-Greedy) sampling policy for varying $n_c$ and $n_s$.
    }
    \label{figure:homogeneous_ss_greedy_time}
\end{figure}

\vspace{-2mm}

\subsection{Heterogeneous Sensor Selection}
\label{section:simulation_02}

In this section, we assume ${n_c}^{(i)} = n_c/K$, ${n_s}^{(i)} = n_s/K$, and $\delta^{(i)} = 1-(1-\delta)^{1/K}$ for the $i$-th partition. We justify our reasoning for ${n_c}^{(i)}$, ${n_s}^{(i)}$, and $\delta^{(i)}$ in Remark~\ref{remark:compare}. We also assume $n_c = 840$ and ${n_p}^{(i)} = 5$.

In Figure~\ref{figure:heterogeneous_ss}a, we plot $\overline{\lambda}( U_H )$, i.e., the upper bound of the steady-state quantity $\overline{\lambda}( P_{\mathcal{H}} )$, for a range of sampled sensors and partition sizes. The value of $\overline{\lambda}( U_H )$ is computed using the sampling distributions, i.e., $\{ {p^{*}}^{(i)} \}_{i \in [K]}$, proposed by the heterogeneous version of \textsc{Algorithm~\ref{alg:04}}. In Figure~\ref{figure:heterogeneous_ss}b, we plot the average time required for each partition to execute the heterogeneous version of \textsc{Algorithm~\ref{alg:04}} for the same range of sampled sensors and partition sizes. Figure~\ref{figure:heterogeneous_ss} shows that a minor degradation in estimation performance (of the upper bound) suggests a modest reduction in execution time.

\begin{figure}
    \centering
    \includegraphics[width=\columnwidth]{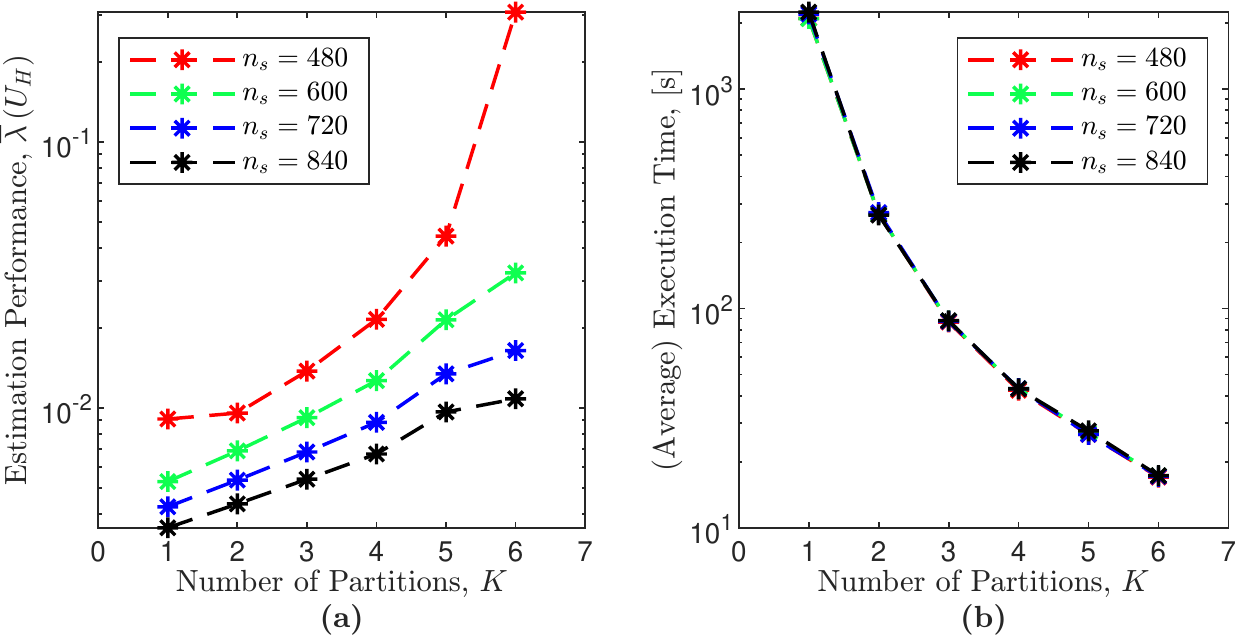}
    \caption{(a) Plot of the steady-state quantity $\overline{\lambda}(U_{\mathcal{H}})$ for varying $n_s$~and~$K$. (b) Plot of the average time required to execute the heterogeneous version of \textsc{Algorithm~\ref{alg:04}} for varying $n_s$ and $K$.
    }
    \label{figure:heterogeneous_ss}
\end{figure}

\vspace{-2mm}

\subsection{Constrained Sensor Selection}
\label{section:simulation_03}

In this section, we assume a uniform set of sampling constraints. Thus, the scalar $k_u$ is sufficient information for fully specifying $k$ according to Remark~\ref{remark:uniform}.

In Figure~\ref{figure:constrained_ss}a, we plot $(\alpha-\delta)$ and $(1-\delta/\alpha)$ to show how their values increase as $\alpha$ approaches unity. Figure~\ref{figure:constrained_ss}a also shows that the probability associated with the CI \eqref{eqn:CI-C2} of Theorem~\ref{theorem:conditional_bounds} is large for a wider range of $\alpha$ values, i.e., for a broader range of sampling constraints, in contrast to the probability associated with the CI \eqref{eqn:CI-C} of Theorem~\ref{theorem:intersection_bounds}.



In Figure~\ref{figure:constrained_ss}b, we plot the parameter $\alpha$ for a range of uniformity factors and sampled sensors. The sampling distributions proposed for each $n_s$ in Section~\ref{section:simulation_01} are used to compute $\alpha$ for the range of $k_u$ and $n_s$ under consideration.

\begin{figure}
    \centering
    \includegraphics[width=\columnwidth]{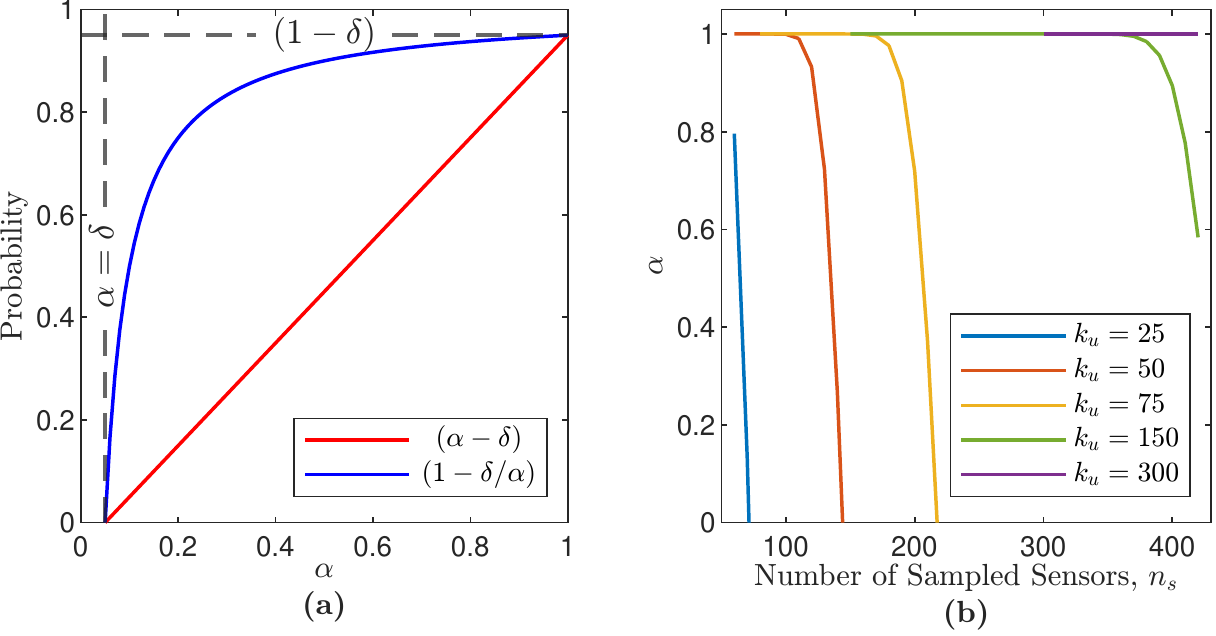}
    \caption{(a) Comparison of the values of $(\alpha-\delta)$ and $(1-\alpha/\delta)$. Note that the former quantities lie on interval $[0,1]$ if $\alpha \in [\delta,1]$. (b) Plot of $\alpha$ (a lower bound on $\mathbb{P}[ \mathcal{K} ]$) for varying $k_u$ and $n_s$.
    }
    \label{figure:constrained_ss}
\end{figure}

\section{Conclusion}
\label{section:conclusion}

In this work, we derived novel CIs for the estimation error covariance arising in three types of stochastic sensor selections and proposed a procedure for finding a sampling distribution that optimizes estimation performance. We also highlighted a few noteworthy properties regarding our CIs, e.g., our semi-definite guarantees as outlined in Remark~\ref{remark:bounds} allow our CIs to bound standard eigenvalue-based metrics irrespective of whether the metrics are submodular or not.
    %
Our proposed sampling distribution was also shown to outperform three other sampling policies: uniform, deterministic greedy, and randomized greedy. 
A noteworthy observation is that the average time required to execute \textsc{Algorithm}~\ref{alg:04} is comparable to greedy approaches for large $n_s/n_c$.


A future direction of research consists of extending our CIs to the class of sampling \textit{without replacement} policies. Other directions of research consist of applying our concentration-based approach to the dual problem, i.e., actuator design, and the joint problem of actuator and sensor selection.

\appendix 
In the section, we provide the complete proof for several claims referenced throughout this paper.

\subsection{Proof of Lemma~\ref{lemma:event_bound}}
In this proof, we show that $\mathbb{P}[ \mathcal{K} ] \geq \Phi ( \alpha )$. We state the following facts for subsequent use in the proof,
\begin{enumerate}[ label=(\roman*) , leftmargin= 8.0mm ]
    \item $\mathbb{P}[ \, \bigcup_{\, i \in [n_c]} \{ I_i \nleq k_i \} \, ] \leq \sum_{\, i \in [n_c]} \mathbb{P}[ I_i \nleq k_i ]$,
    \item $\mathbb{P}[ \, I_i \nleq k_i \, ] = 1 - \mathbb{P}[ \, I_i \leq k_i \, ]$,
    \item $\mathbb{P}[ \, I_i \leq k_i \, ] = \textstyle\sum_{j=0}^{k_i} \mathbb{P}[ \, I_i = j \, ]$,
\end{enumerate}
where fact~(i) holds for any union of events according to Boole's inequality, fact~(ii) is a property of complements, and fact~(iii) holds since the event $\{ I_i \leq k_i \}$ is equivalent to $\textstyle\bigcup_{j=0}^{k_i} \{ I_i = j \}$, a union of mutually exclusive events. 

First, we derive the inequality $\mathbb{P}[ \mathcal{K}^{c} ] \leq \beta$, i.e.,
\begin{gather}
\mathbb{P}[ \mathcal{K}^{c} ] 
= \mathbb{P}[ \, \textstyle\bigcup_{i=1}^{n_c} \{ I_i \nleq k_i \} \, ]
\overset{(a)}{\leq} \textstyle\sum_{i=1}^{n_c} \mathbb{P}[ I_i \nleq k_i ] \nonumber \\
\label{eqn:inequality_beta}
\overset{(b)}{=} n_c - \textstyle\sum_{i=1}^{n_c} \textstyle\sum_{j=0}^{k_i} f_5 ( n_s, j , p_i ) =: \beta,
\end{gather}
where step~(a) employs fact~(i), and step~(b) employs the following, (ii), (iii), and \eqref{eqn:C}, in that order. Note that $\beta \geq 0$ since $\mathbb{P}[ \mathcal{K}^{c} ] \in [0,1]$. We finally derive the inequality $\mathbb{P}[ \mathcal{K} ] \geq \Phi ( \alpha )$ by applying \eqref{eqn:inequality_beta} to the true statement $\mathbb{P}[ \mathcal{K} ] = 1 - \mathbb{P}[ \mathcal{K}^{c} ]$, i.e., $\mathbb{P}[ \mathcal{K} ] \geq (1-\beta) =: \alpha$. Note that $\alpha \leq 1$ since $\beta \geq 0$.
\qed

\subsection{Proof of Theorem \ref{theorem:SS_bounds_t}}
In this proof, we derive the CI for the filtered error covariance of a homogeneous selection.

First, the statistical properties of the covariance equations of the Kalman filter are stated for a fixed selection of sensors, i.e., a deterministic selection, in Lemma~\ref{lemma:KF_01}. Let $C$ and $R$ denote the output and measurement noise covariance matrix, resp., for a fixed selection. Refer to Section~4 of Chapter~5 of \cite{anderson2012optimal} or Theorem~23 of \cite{simon2006optimal} for a generalization of Lemma~\ref{lemma:KF_01}. Following the lemma, we formulate two corollaries for subsequent use in the proof. Corollary~\ref{corollary:KF_00} is a special case of Lemma~\ref{lemma:KF_01} and Corollary~\ref{eqn:corollary:KF_01} is an alternate formulation of Corollary~\ref{corollary:KF_00}.

\begin{lemma}
\label{lemma:KF_01}
Assume $Q \in \mathbb{S}_{+}^{m}$, $C \in \mathbb{R}^{n_s \times m}$, and $R \in \mathbb{S}_{++}^{n_s}$. Let $P_{(t)}$ and $ \Sigma_{(t+1)}$ denote the filtered and predicted error covariance at time instants $t$ and $t+1$, resp., where $P_{(t)} = f_1( \hspace{0.25mm} \Sigma_{(t)} , C , R \hspace{0.25mm} )$ and $\Sigma_{(t+1)} = f_4( \hspace{0.25mm} P_{(t)} \hspace{0.25mm} )$. If $(A,Q^{1/2})$ is stabilizable and $(A,C)$ is detectable, then the filtered error covariance $P_{(t)}\in \mathbb{S}_{+}^{m}$ converges to the unique p.s.d. steady-state solution $P \in \mathbb{S}_{+}^{m}$ for any $\Sigma_{(0)} \in \mathbb{S}_{+}^{m}$.
\end{lemma}

\begin{corollary}
\label{corollary:KF_00}
Assume $Q \in \mathbb{S}_{++}^{m}$. Let $P_{(t)}$ and $ \Sigma_{(t+1)}$ denote the filtered and predicted error covariance at time instants $t$ and $t+1$, resp., where $P_{(t)} = f_3( \hspace{0.25mm} \Sigma_{(t)} , C , R \hspace{0.25mm} )$ and $\Sigma_{(t+1)} = f_4( \hspace{0.25mm} P_{(t)} \hspace{0.25mm} )$. If $(A,C)$ is detectable, then the filtered error covariance $P_{(t)} \in \mathbb{S}_{++}^{m}$ converges to the unique p.d. steady-state solution $P \in \mathbb{S}_{++}^{m}$ for any $\Sigma_{(0)} \in \mathbb{S}_{++}^{m}$.
\end{corollary}

\begin{proof}
By assuming $Q$ and $\Sigma_{(0)}$ are elements of the set $\mathbb{S}_{++}^{m}$, Lemma~\ref{lemma:KF_01} changes accordingly: (i) the stabilizability condition is satisfied, (ii) the quantity $f_1( \hspace{0.25mm} \Sigma_{(t)} , C , R \hspace{0.25mm} )$ is equivalent to $f_3( \hspace{0.25mm} \Sigma_{(t)} , C , R \hspace{0.25mm} )$ due to the matrix inversion lemma, and (iii) the steady-state solution of $P_{(t)}$ is p.d.
\end{proof}

\begin{corollary}
\label{eqn:corollary:KF_01}
Assume $(A,C)$ is detectable. Suppose $Q \in \mathbb{S}_{++}^{m}$. Then, the covariance equation $P_{(t+1)} = f_2( \hspace{0.25mm} P_{(t)} , C^{T} R^{-1} C \hspace{0.25mm} )$ of the filtered error covariance matrix converges to the unique p.d. steady-state solution $P \in \mathbb{S}_{++}^{m}$ for any $P_{(0)} \in \mathbb{S}_{++}^{m}$.
\end{corollary}

\begin{proof}
First, we merge the predicted and filtered covariance matrices of Corollary~\ref{corollary:KF_00} at time instant $t+1$ to derive the equality $P_{(t+1)} = f_3( \hspace{0.25mm} f_4( \hspace{0.25mm} P_{(t)} \hspace{0.25mm} ) , C , R \hspace{0.25mm} ) = f_2( \hspace{0.25mm} P_{(t)} , C^{T} R^{-1} C \hspace{0.25mm} )$. Next, observe that $P_{(0)} \in \mathbb{S}_{++}^{m}$ for any $\Sigma_{(0)} \in \mathbb{S}_{++}^{m}$ since $P_{(t)} = f_3( \hspace{0.25mm} \Sigma_{(t)} , C , R \hspace{0.25mm} )$. Finally, the fact that the covariance equation of Corollary~\ref{eqn:corollary:KF_01} converges to a unique p.d. solution directly follows from Corollary~\ref{corollary:KF_00} and omitted for brevity.
\end{proof}

Next, we formulate a lemma and a corollary for subsequent use in the proof. In Lemma~\ref{lemma:KF_03}, we state the statistical properties of two Kalman filters, where the quantities of the first and second filter are denoted by a subscript of $1$ and $2$, resp. Lemma~\ref{lemma:KF_03} directly follows from Corollary~\ref{corollary:KF_00}. In Corollary~\ref{corollary:KF_01}, a corollary of Lemma~\ref{lemma:KF_03}, we compare the filtered and predicted error covariance of the two filters in the semi-definite sense.

\begin{lemma}
\label{lemma:KF_03}
Assume $Q \in \mathbb{S}_{++}^{m}$, $C_i \in \mathbb{R}^{n_s \times m}$, and $R_{i} \in \mathbb{S}_{++}^{n_s}$. Let $P_{i,(t)}$ and $\Sigma_{i,(t+1)}$ denote the filtered and predicted error covariance of the Kalman filter $i \in \{1,2\}$ at time instants~$t$ and $t+1$, resp., where $P_{i,(t)} = f_3( \hspace{0.25mm} \Sigma_{i,(t)} , C_{i} , R_{i} \hspace{0.25mm} )$ and $\Sigma_{i,(t+1)} = f_4( \hspace{0.25mm} P_{i,(t)} \hspace{0.25mm} )$. If $(A,C_i)$ is detectable, then the filtered error covariance $P_{i,(t)} \in \mathbb{S}_{++}^{m}$ converges to the unique p.d. steady-state solution $P_{i} \in \mathbb{S}_{++}^{m}$ for any $\Sigma_{i,(0)} \in \mathbb{S}_{++}^{m}$.
\end{lemma}

\begin{corollary}
\label{corollary:KF_01}
If the following conditions, $Q, \Sigma_{1,(t)} \in \mathbb{S}_{++}^{m}$,
\begin{gather}
\label{eqn:c1}
\Sigma_{1,(t)} \preceq \Sigma_{2,(t)}, \
C_{2}^{T} R_{2}^{-1} C_{2}^{} \preceq C_{1}^{T} R_{1}^{-1} C_{1}^{},
\end{gather}
hold for any time instant $t$, then $P_{1,(t)}, \Sigma_{1,(t+1)} \in \mathbb{S}_{++}^{m}$,
\begin{gather}
\label{eqn:c5}
P_{1,(t)} \preceq P_{2,(t)}, \ \Sigma_{1,(t+1)} \preceq \Sigma_{2,(t+1)}.
\end{gather}
\end{corollary}

\begin{proof}
In this proof, we show that the conditions outlined in \eqref{eqn:c1} imply the inequalities of \eqref{eqn:c5} for any time instant $t$.

First, we show that the former inequality of \eqref{eqn:c1} implies the former inequality of \eqref{eqn:c5} for any time instant $t$, i.e.,
\begin{gather}
\Sigma_{1,(t)} \preceq \Sigma_{2,(t)} \overset{(a)}{\Leftrightarrow} \Sigma_{2,(t)}^{-1} \preceq \Sigma_{1,(t)}^{-1}   \nonumber \\
\overset{(b)}{\Leftrightarrow} \Sigma_{2,(t)}^{-1} + C_{2}^{T} R_{2}^{-1} C_{2}^{} \preceq \Sigma_{1,(t)}^{-1} + C_{1}^{T} R_{1}^{-1} C_{1}^{}   \nonumber \\
\overset{(c)}{\Leftrightarrow} P_{2,(t)}^{-1} \preceq P_{1,(t)}^{-1} \Leftrightarrow P_{1,(t)} \preceq P_{2,(t)},   \nonumber
\end{gather}
where step~(a) assumes $\Sigma_{1,(t)} \in \mathbb{S}_{++}^{m}$, step~(b) employs the latter inequality of \eqref{eqn:c1}, and step~(c) employs the filtered error covariance equation of Lemma~\ref{lemma:KF_03}. Next, we show that the former inequality of \eqref{eqn:c5} implies the latter, i.e.,
\begin{gather}
P_{1,(t)} \preceq P_{2,(t)} \overset{(d)}{\Rightarrow} A P_{1,(t)} A^T \preceq A P_{2,(t)} A^T   \nonumber \\
\Leftrightarrow A P_{1,(t)} A^T \! \! + \! Q \preceq A P_{2,(t)} A^T \! \! + \! Q \, \overset{(e)}{\Leftrightarrow} \, \Sigma_{1,(t+1)} \preceq \Sigma_{2,(t+1)},   \nonumber
\end{gather}
where step~(d) holds by a property of Hermitian matrices, explicitly referred to as the Conjugation Rule in~\cite{tropp2015introduction}, step~(e) assumes $Q \in \mathbb{S}_{++}^{m}$, and step~(f) employs the predicted error covariance equation of Lemma~\ref{lemma:KF_03}. Note that the condition $P_{1,(t)}, \Sigma_{1,(t+1)} \in \mathbb{S}_{++}^{m}$ trivially follows from the above derivation. Finally, the former and latter inequality of \eqref{eqn:c5} can be derived for subsequent time instants by induction.
\end{proof}


Next, we show how Corollary~\ref{corollary:KF_01} applies to three Kalman filters. If a third filter, denoted by subscript~$3$, is introduced, then according to Corollary~\ref{corollary:KF_01} the following assumptions,
\begin{enumerate}[ label=(\roman*) , leftmargin= 8.0mm ]
    \item $\Sigma_{1,(t)} \in \mathbb{S}_{++}^{m}$, $\Sigma_{1,(t)} \preceq \Sigma_{2,(t)} \preceq \Sigma_{3,(t)}$,
    \item $C_{3}^{T} R_{3}^{-1} C_{3}^{} \preceq C_{2}^{T} R_{2}^{-1} C_{2}^{} \preceq C_{1}^{T} R_{1}^{-1} C_{1}^{}$,
\end{enumerate}
imply $P_{1,(t)} \in \mathbb{S}_{++}^{m}$ and $P_{1,(t)} \preceq P_{2,(t)} \preceq P_{3,(t)}$. Note that we do not require the assumption $Q \in \mathbb{S}_{++}^{m}$ in the above statement since we only invoke the former inequality of \eqref{eqn:c5} of Corollary~\ref{corollary:KF_01}. We define the quantities $\Phi$, $C_{-}$, and $C_{+}$ for subsequent use in the proof, i.e., $\Phi := n_s^{1/2} \, \mathbb{E}[Z]^{1/2}$,
\begin{align}
\label{eqn:C_def}
C_{-} := (1-\epsilon)^{1/2} \, \Phi, \  C_{+} := (1+\epsilon)^{1/2} \, \Phi.
\end{align}

Next, we derive the assumptions and guarantees of Theorem~\ref{theorem:SS_bounds_t}. If we employ the following substitutions:
\begin{itemize}
    \item $R_1$, $R_2$, $R_3$ for $I_m$, $R_{\mathcal{S}}$, $I_m$, resp.,
    \item $C_1$, $C_2$, $C_3$ for $C_{+}$, $C_{\mathcal{S}}$, $C_{-}$, resp.,
    \item $\Sigma_{1,(t)}$, $\Sigma_{2,(t)}$, $\Sigma_{3,(t)}$ for $\Sigma_{L,(t)}$, $\Sigma_{\mathcal{S},(t)}$, $\Sigma_{U,(t)}$, resp.,
    \item $P_{1,(t)}$, $P_{2,(t)}$, $P_{3,(t)}$ for $P_{L,(t)}$, $P_{\mathcal{S},(t)}$, $P_{U,(t)}$, resp.,
\end{itemize}
then the assumptions of the former paragraph simplify to
\begin{enumerate}[ label=(\roman*) , leftmargin= 8.0mm ]
    \item $\Sigma_{L,(t)} \in \mathbb{S}_{++}^{m}$, $\Sigma_{L,(t)} \preceq \Sigma_{\mathcal{S},(t)} \preceq \Sigma_{U,(t)}$,
    \item $(1-\epsilon) \, n_s \, \mathbb{E}[Z] \, \preceq \, C_{\mathcal{S}}^{T} R_{\mathcal{S}}^{-1} C_{\mathcal{S}}^{} \, \preceq \, (1+\epsilon) \, n_s \, \mathbb{E}[Z]$,
\end{enumerate}
resp., and the corresponding guarantees of the former paragraph reduce to the following, $P_{L,(t)} \in \mathbb{S}_{++}^{m}$ and $P_{L,(t)} \preceq P_{\mathcal{S},(t)} \preceq P_{U,(t)}$, resp. We complete the proof by highlighting several observations. First, we satisfy the condition~(i) by assuming the predicted error covariance matrices for all three filters are equal. We denote them as $\Sigma_{(t)}$ in Theorem~\ref{theorem:SS_bounds_t} for convenience. Second, observe that the condition~(ii) cannot be guaranteed to hold in a deterministic sense since $C_{\mathcal{S}}^{T} R_{\mathcal{S}}^{-1} C_{\mathcal{S}}^{}$ is a random quantity. However, it can be guaranteed to hold in a probabilistic sense by applying Theorem~\ref{theorem:AW}. Recall from Section~\ref{subsection:ss_KF} that the quantity $C_{\mathcal{S}}^{T} R_{\mathcal{S}}^{-1} C_{\mathcal{S}}^{}$ is equivalent to a sum of i.i.d. random p.s.d. matrices. As a consequence, the event $\{ P_{L,(t)} \preceq P_{\mathcal{S},(t)} \preceq P_{U,(t)} \}$ also holds in a probabilistic sense. We denote $P_{U,(t)}$ and $P_{L,(t)}$ as $U_{S,(t)}$ and $ L_{S,(t)}$, resp., to clarify that the concentration bounds of Theorem~\ref{theorem:SS_bounds_t} correspond to a homogeneous selection.
\qed



\subsection{Proof of Theorem~\ref{theorem:SS_bounds}}

In this proof, we derive the CI for the steady-state error covariance of a homogeneous selection.

First, we derive three corollaries for subsequent use in the proof. Corollary~\ref{corollary:KF_02} directly follows from Corollary~\ref{corollary:KF_01}.

\begin{corollary}
\label{corollary:KF_02}
Assume $Q, \Sigma_{1,(0)} \in \mathbb{S}_{++}^{m}$, $\Sigma_{1,(0)} \preceq \Sigma_{2,(0)}$, and $C_{2}^{T} R_{2}^{-1} C_{2}^{} \preceq C_{1}^{T} R_{1}^{-1} C_{1}^{}$. Then, the following, $P_{1,(t)} \in \mathbb{S}_{++}^{m}$ and $P_{1,(t)} \preceq P_{2,(t)}$, hold for any time instant $t \geq 0$.
\end{corollary}

If the pairs $(A,C_1)$ and $(A,C_2)$ are detectable, then the matrices $P_{1,(t)}$ and $P_{2,(t)}$ of Corollary~\ref{corollary:KF_02} converge to the unique p.d. steady-state solutions $P_1$ and $P_2$, resp., as a consequence of Lemma~\ref{lemma:KF_03}. The steady-state solutions satisfy the inequality $P_{1} \preceq P_{2}$ as a consequence of Corollary~\ref{corollary:KF_02}. Since the steady-state solutions $P_1$ and $P_2$ are unique for any matrices $\Sigma_{1,(0)}$ and $\Sigma_{2,(0)}$ in the set $\mathbb{S}_{++}^{m}$, resp., the assumption $\Sigma_{1,(0)} \preceq \Sigma_{2,(0)}$ of Corollary~\ref{corollary:KF_02} is not required to guarantee $P_{1} \preceq P_{2}$. We summarize the conclusions of this paragraph in Corollary~\ref{corollary:KF_03}. By following a proof similar to that of Corollary~\ref{corollary:KF_01}, we obtain Corollary~\ref{corollary:KF_04}, an alternate formulation of Corollary~\ref{corollary:KF_03}.

\begin{corollary}
\label{corollary:KF_03}
Assume $(A,C_1)$ and $(A,C_2)$ are detectable. Suppose $Q \in \mathbb{S}_{++}^{m}$, $\Sigma_{1,(0)}, \Sigma_{2,(0)} \in \mathbb{S}_{++}^{m}$, and $C_{2}^{T} R_{2}^{-1} C_{2}^{} \preceq C_{1}^{T} R_{1}^{-1} C_{1}^{}$. Then, $P_{1}  \in \mathbb{S}_{++}^{m}$ and $P_{1} \preceq P_{2}$.
\end{corollary}

\begin{corollary}
\label{corollary:KF_04}
Assume $(A,C_1)$ and $(A,C_2)$ are detectable. Suppose $Q \in \mathbb{S}_{++}^{m}$. Let $P_1$ and $P_2$ denote the unique p.d. steady-state solution to the following,
\begin{align*}
P_{1,(t+1)} = f_2( \, P_{1,(t)} \, , \, C_{1}^{T} R_{1}^{-1} C_{1} \, ), \\
P_{2,(t+1)} = f_2( \, P_{2,(t)} \, , \, C_{2}^{T} R_{2}^{-1} C_{2} \, )
\end{align*}
resp., such that $P_{1,(0)}, P_{2,(0)} \in \mathbb{S}_{++}^{m}$. Assume $C_{2}^{T} R_{2}^{-1} C_{2}^{} \preceq C_{1}^{T} R_{1}^{-1} C_{1}^{}$. Then, $P_{1} \in \mathbb{S}_{++}^{m}$ and $P_{1} \preceq P_{2}$.
\end{corollary}

Next, we show how Corollary~\ref{corollary:KF_04} applies to three Kalman filters. We assume $Q \in \mathbb{S}_{++}^{m}$ from this point onward in the proof. If a third filter, denoted by subscript~$3$, is introduced, then according to Corollary~\ref{corollary:KF_04} the following assumptions,
\begin{enumerate}[ label=(\roman*) , leftmargin= 8.0mm ]
    \item $(A,C_1)$, $(A,C_2)$, and $(A,C_3)$ are detectable,
    \item $C_{3}^{T} R_{3}^{-1} C_{3}^{} \preceq C_{2}^{T} R_{2}^{-1} C_{2}^{} \preceq C_{1}^{T} R_{1}^{-1} C_{1}^{}$,
    \item $P_{1,(0)}, P_{2,(0)}, P_{3,(0)} \in \mathbb{S}_{++}^{m}$,
\end{enumerate}
imply $P_{1} \in \mathbb{S}_{++}^{m}$ and $P_{1} \preceq P_{2} \preceq P_{3}$.

Next, we derive the assumptions and guarantees of Theorem~\ref{theorem:SS_bounds}. If we employ the following substitutions:
\begin{itemize}
    \item $R_1$, $R_2$, $R_3$ for $I_m$, $R_{\mathcal{S}}$, $I_m$, resp.,
    \item $C_1$, $C_2$, $C_3$ for $C_{+}$, $C_{\mathcal{S}}$, $C_{-}$, resp.,
    \item $P_{1}$, $P_{2}$, $P_{3}$ for $P_{L}$, $P_{\mathcal{S}}$, $P_{U}$, resp.,
\end{itemize}
then the assumptions of the former paragraph simplify to
\begin{enumerate}[ label=(\roman*) , leftmargin= 8.0mm ]
    \item $(A,C_{-})$, $(A,C_{\mathcal{S}})$, and $(A,C_{+})$ are detectable,
    \item $(1-\epsilon) \, n_s \, \mathbb{E}[Z] \, \preceq \, C_{\mathcal{S}}^{T} R_{\mathcal{S}}^{-1} C_{\mathcal{S}}^{} \, \preceq \, (1+\epsilon) \, n_s \, \mathbb{E}[Z]$,
    \item $P_{L,(0)}, P_{\mathcal{S},(0)}, P_{U,(0)} \in \mathbb{S}_{++}^{m}$,
\end{enumerate}
resp., and the corresponding guarantees of the former paragraph reduce to the following, $P_{L} \in \mathbb{S}_{++}^{m}$ and $P_{L} \preceq P_{\mathcal{S}} \preceq P_{U}$, resp. Recall that $C_{-}$ and $C_{+}$ are defined in \eqref{eqn:C_def}. We complete the proof by highlighting several observations. First, observe that the pairs $(A,C_{-})$ and $(A,C_{+})$ are detectable if and only if the pair $(A,\mathbb{E}[Z]^{1/2})$ is detectable. We state the latter detectability condition over the former in Theorem~\ref{theorem:SS_bounds}. Second, observe that the condition~(ii) cannot be guaranteed to hold in a deterministic sense since $C_{\mathcal{S}}^{T} R_{\mathcal{S}}^{-1} C_{\mathcal{S}}^{}$ is a random quantity. However, it can be guaranteed to hold in a probabilistic sense by applying Theorem~\ref{theorem:AW}. Recall from Section~\ref{subsection:ss_KF} that the quantity $C_{\mathcal{S}}^{T} R_{\mathcal{S}}^{-1} C_{\mathcal{S}}^{}$ is equivalent to a sum of i.i.d. random p.s.d. matrices. As a consequence, the event $\{ P_{L} \preceq P_{\mathcal{S}} \preceq P_{U} \}$ also holds in a probabilistic sense. We denote $P_{U}$ and $P_{L}$ as $U_{S}$ and $ L_{S}$, resp., to clarify that the concentration bounds of Theorem~\ref{theorem:SS_bounds} correspond to a homogeneous selection.
\qed






\bibliographystyle{IEEEtran} 
\bibliography{references}

\begin{thebibliography}{10}
\providecommand{\url}[1]{#1}
\csname url@samestyle\endcsname
\providecommand{\newblock}{\relax}
\providecommand{\bibinfo}[2]{#2}
\providecommand{\BIBentrySTDinterwordspacing}{\spaceskip=0pt\relax}
\providecommand{\BIBentryALTinterwordstretchfactor}{4}
\providecommand{\BIBentryALTinterwordspacing}{\spaceskip=\fontdimen2\font plus
\BIBentryALTinterwordstretchfactor\fontdimen3\font minus
  \fontdimen4\font\relax}
\providecommand{\BIBforeignlanguage}[2]{{%
\expandafter\ifx\csname l@#1\endcsname\relax
\typeout{** WARNING: IEEEtran.bst: No hyphenation pattern has been}%
\typeout{** loaded for the language `#1'. Using the pattern for}%
\typeout{** the default language instead.}%
\else
\language=\csname l@#1\endcsname
\fi
#2}}
\providecommand{\BIBdecl}{\relax}
\BIBdecl

\bibitem{van2001}
M.~Van De~Wal and B.~De~Jager, ``A review of methods for input/output
  selection,'' \emph{Automatica}, vol.~37, no.~4, pp. 487--510, 2001.

\bibitem{taha2018time}
A.~F. Taha, N.~Gatsis, T.~Summers, and S.~A. Nugroho, ``Time-varying sensor and
  actuator selection for uncertain cyber-physical systems,'' \emph{IEEE
  Transactions on Control of Network Systems}, vol.~6, no.~2, pp. 750--762,
  2018.

\bibitem{siami2020deterministic}
M.~Siami, A.~Olshevsky, and A.~Jadbabaie, ``Deterministic and randomized
  actuator scheduling with guaranteed performance bounds,'' \emph{IEEE
  Transactions on Automatic Control}, vol.~66, no.~4, pp. 1686--1701, 2020.

\bibitem{krause2014submodular}
A.~Krause and D.~Golovin, ``Submodular function maximization,''
  \emph{Tractability}, vol.~3, pp. 71--104, 2014.

\bibitem{iwata2008submodular}
S.~Iwata, ``Submodular function minimization,'' \emph{Mathematical
  Programming}, vol. 112, no.~1, pp. 45--64, 2008.

\bibitem{clark2016submodularity}
A.~Clark, B.~Alomair, L.~Bushnell, and R.~Poovendran, \emph{Submodularity in
  dynamics and control of networked systems}.\hskip 1em plus 0.5em minus
  0.4em\relax Springer, 2016.

\bibitem{jawaid2015submodularity}
S.~T. Jawaid and S.~L. Smith, ``Submodularity and greedy algorithms in sensor
  scheduling for linear dynamical systems,'' \emph{Automatica}, vol.~61, pp.
  282--288, 2015.

\bibitem{zhang2017sensor}
H.~Zhang, R.~Ayoub, and S.~Sundaram, ``{Sensor selection for Kalman filtering
  of linear dynamical systems: Complexity, limitations and greedy
  algorithms},'' \emph{Automatica}, vol.~78, pp. 202--210, 2017.

\bibitem{nemhauser1978analysis}
G.~L. Nemhauser, L.~A. Wolsey, and M.~L. Fisher, ``{An analysis of
  approximations for maximizing submodular set functions—I},''
  \emph{Mathematical programming}, vol.~14, no.~1, pp. 265--294, 1978.

\bibitem{kohara2020sensor}
A.~Kohara, K.~Okano, K.~Hirata, and Y.~Nakamura, ``Sensor placement minimizing
  the state estimation mean square error: Performance guarantees of greedy
  solutions,'' in \emph{2020 59th IEEE Conference on Decision and Control
  (CDC)}.\hskip 1em plus 0.5em minus 0.4em\relax IEEE, 2020, pp. 1706--1711.

\bibitem{hartman2019near}
D.~Hartman and J.~S. Baras, ``{Near-optimal solution to the non-uniform
  sampling problem in Kalman filtering},'' in \emph{2019 IEEE 58th Conference
  on Decision and Control (CDC)}.\hskip 1em plus 0.5em minus 0.4em\relax IEEE,
  2019, pp. 6404--6411.

\bibitem{chamon2020approximate}
L.~F. Chamon, G.~J. Pappas, and A.~Ribeiro, ``{Approximate supermodularity of
  Kalman filter sensor selection},'' \emph{IEEE Transactions on Automatic
  Control}, vol.~66, no.~1, pp. 49--63, 2020.

\bibitem{hashemi2021randomized}
A.~Hashemi, M.~Ghasemi, H.~Vikalo, and U.~Topcu, ``Randomized greedy sensor
  selection: Leveraging weak submodularity,'' \emph{IEEE Transactions on
  Automatic Control}, vol.~66, no.~1, pp. 199--212, 2021.

\bibitem{mirzasoleiman2015lazier}
B.~Mirzasoleiman, A.~Badanidiyuru, A.~Karbasi, J.~Vondr{\'a}k, and A.~Krause,
  ``Lazier than lazy greedy,'' in \emph{Proceedings of the AAAI Conference on
  Artificial Intelligence}, vol.~29, no.~1, 2015.

\bibitem{han2021stochastic}
L.~Han, M.~Li, D.~Xu, and D.~Zhang, ``Stochastic-lazier-greedy algorithm for
  monotone non-submodular maximization,'' \emph{Journal of Industrial \&
  Management Optimization}, vol.~17, no.~5, p. 2607, 2021.

\bibitem{bopardikar2019randomized}
S.~D. Bopardikar, O.~Ennasr, and X.~Tan, ``Randomized sensor selection for
  nonlinear systems with application to target localization,'' \emph{IEEE
  Robotics and Automation Letters}, vol.~4, no.~4, pp. 3553--3560, 2019.

\bibitem{bopardikar2021randomized}
S.~D. Bopardikar, ``A randomized approach to sensor placement with
  observability assurance,'' \emph{Automatica}, vol. 123, p. 109340, 2021.

\bibitem{ahlswede2002strong}
R.~Ahlswede and A.~Winter, ``Strong converse for identification via quantum
  channels,'' \emph{IEEE Transactions on Information Theory}, vol.~48, no.~3,
  pp. 569--579, 2002.

\bibitem{calle2021probabilistic}
C.~I. Calle and S.~D. Bopardikar, ``{Probabilistic performance bounds for
  randomized sensor selection in Kalman filtering},'' in \emph{2021 American
  Control Conference (ACC)}.\hskip 1em plus 0.5em minus 0.4em\relax IEEE, 2021,
  pp. 4395--4400.

\bibitem{simon2006optimal}
D.~Simon, \emph{Optimal state estimation: Kalman, H infinity, and nonlinear
  approaches}.\hskip 1em plus 0.5em minus 0.4em\relax John Wiley \& Sons, 2006.

\bibitem{qiu2014sums}
R.~Qiu and M.~Wicks, ``Sums of matrix-valued random variables,'' in
  \emph{Cognitive Networked Sensing and Big Data}.\hskip 1em plus 0.5em minus
  0.4em\relax Springer, 2014, pp. 85--144.

\bibitem{vandenberghe2005interior}
L.~Vandenberghe, V.~R. Balakrishnan, R.~Wallin, A.~Hansson, and T.~Roh,
  ``Interior-point algorithms for semidefinite programming problems derived
  from the kyp lemma,'' \emph{Positive polynomials in control}, pp. 195--238,
  2005.

\bibitem{anderson2012optimal}
B.~D. Anderson and J.~B. Moore, \emph{Optimal filtering}.\hskip 1em plus 0.5em
  minus 0.4em\relax Courier Corporation, 2012.

\bibitem{tropp2015introduction}
J.~A. Tropp, ``An introduction to matrix concentration inequalities,''
  \emph{Foundations and Trends{\textregistered} in Machine Learning}, vol.~8,
  no. 1-2, pp. 1--230, 2015.

\end{thebibliography}

\end{document}